\documentclass[11pt]{article}

\usepackage{algorithm}
\usepackage{algorithmic}
\usepackage{amsfonts}
\usepackage{amsmath}
\usepackage{amssymb}
\usepackage{amsthm}
\usepackage{bbm}
\usepackage{bm}
\usepackage{boxedminipage}
\usepackage{caption}
\usepackage{enumitem}
\usepackage{float}
\usepackage{graphicx}
\usepackage{ifthen}
\usepackage[utf8]{inputenc} 
\usepackage{listings}
\usepackage{mathtools}
\usepackage{multirow}
\usepackage{nicefrac}
\usepackage{relsize}
\usepackage{subcaption}
\usepackage{xparse}
\usepackage{xspace}
\usepackage[usenames,dvipsnames,svgnames,table]{xcolor}

\usepackage{todonotes}
\usepackage{textcase}
\usepackage{soul}

\usepackage[colorlinks=true, linkcolor=red, urlcolor=blue, citecolor=gray]{hyperref}
\usepackage[margin=1in]{geometry}
\usepackage{tikz}

\makeatletter
\def\hlinewd#1{%
  \noalign{\ifnum0=`}\fi\hrule \@height #1 \futurelet
  \reserved@a\@xhline}
\makeatother

\newtheorem*{rep@theorem}{\rep@title}
\newcommand{\newreptheorem}[2]{%
\newenvironment{rep#1}[1]{%
 \def\rep@title{\theoremref{##1} Restated}%
 \begin{rep@theorem}}%
 {\end{rep@theorem}}}
\makeatother
\makeatletter
\newtheorem*{rep@lemma}{\rep@title}
\newcommand{\newreplemma}[2]{%
\newenvironment{rep#1}[1]{%
 \def\rep@title{\lemmaref{##1} Restated}%
 \begin{rep@lemma}}%
 {\end{rep@lemma}}}
\makeatother
\makeatletter
\newtheorem*{rep@claim}{\rep@title}
\newcommand{\newrepclaim}[2]{%
\newenvironment{rep#1}[1]{%
 \def\rep@title{\claimref{##1} Restated}%
 \begin{rep@claim}}%
 {\end{rep@claim}}}
\makeatother
\makeatletter
\newtheorem*{rep@problem}{\rep@title}
\newcommand{\newrepproblem}[2]{%
\newenvironment{rep#1}[1]{%
 \def\rep@title{\problemref{##1} Restated}%
 \begin{rep@problem}}%
 {\end{rep@problem}}}
\makeatother
\newtheorem{theorem}{Theorem}
\newtheorem{definition}{Definition}

\newtheorem{importedtheorem}[theorem]{Imported Theorem}

\newtheorem{lemma}[theorem]{Lemma}
\newtheorem{importedlemma}[theorem]{Imported Lemma}

\newtheorem{problem}[definition]{Problem}

\newreplemma{lemma}{Lemma}
\newreptheorem{theorem}{Theorem}
\newrepproblem{problem}{Problem}

\newcommand{\namedref}[2]{\texorpdfstring{\hyperref[#2]{#1~\ref*{#2}}}{#1~\ref*{#2}}\xspace}
\newcommand{\lemmaref}[1]{\namedref{Lemma}{lem:#1}}
\newcommand{\theoremref}[1]{\namedref{Theorem}{thm:#1}}
\newcommand{\claimref}[1]{\namedref{Claim}{clm:#1}}

\newcommand{\figureref}[1]{\namedref{Figure}{fig:#1}}
\newcommand{\equationref}[1]{\namedref{Equation}{eq:#1}}

\newcommand{\problemref}[1]{\namedref{Problem}{prob:#1}}
\newcommand{\algorithmref}[1]{\namedref{Algorithm}{alg:#1}}
\newcommand{\importedtheoremref}[1]{\namedref{Imported Theorem}{impthm:#1}}
\newcommand{\importedlemmaref}[1]{\namedref{Imported Lemma}{implem:#1}}

\newcommand{\sectionref}[1]{\namedref{Section}{sec:#1}}
\newcommand{\appendixref}[1]{\namedref{Appendix}{app:#1}}

\newcommand{\abs}[1]{\ensuremath{\vert{#1}\vert}\xspace}
\newcommand{\abslr}[1]{\ensuremath{\left\vert{#1}\right\vert}\xspace}

\newcommand{\eps}[0]{\ensuremath{\varepsilon}}
\let\epsilon\eps

\newcommand{\cA}{\ensuremath{{\mathcal A}}\xspace}

\newcommand{\cN}{\ensuremath{{\mathcal N}}\xspace}

\newcommand{\cP}{\ensuremath{{\mathcal P}}\xspace}
\newcommand{\cQ}{\ensuremath{{\mathcal Q}}\xspace}

\newcommand{\cX}{\ensuremath{{\mathcal X}}\xspace}

\newcommand{\bbN}{\ensuremath{{\mathbb N}}\xspace}

\newcommand{\bbR}{\ensuremath{{\mathbb R}}\xspace}

\newcommand{\defeq}[0]{\ensuremath{\;{\vcentcolon=}\;}\xspace}

\DeclareMathOperator*{\E}{\mathbb{E}}

\newcommand{\iid}[0]{\text{i.i.d.}\xspace}

\DeclareMathOperator{\diag}{diag}
\DeclareMathOperator{\trace}{tr}

\DeclareMathOperator{\rank}{rank}

\DeclareMathOperator*{\Var}{Var}

\newcommand{\eqdist}{\mathbin{\stackrel{\rm dist}{=}}}

\newcommand{\mat}[1]{\boldsymbol{#1}}
\renewcommand{\vec}[1]{\boldsymbol{\mathrm{#1}}}
\newcommand{\vecalt}[1]{\boldsymbol{#1}}

\newcommand{\normof}[1]{\|#1\|}

\newcommand{\bmat}[1]{\begin{bmatrix} #1 \end{bmatrix}}
\newcommand{\spmat}[1]{\left(\begin{smallmatrix} #1 \end{smallmatrix}\right)}

\newcommand{\eye}{\mat{I}\xspace}
\newcommand{\mA}{\ensuremath{\mat{A}}\xspace}
\newcommand{\mB}{\ensuremath{\mat{B}}\xspace}
\newcommand{\mC}{\ensuremath{\mat{C}}\xspace}
\newcommand{\mD}{\ensuremath{\mat{D}}\xspace}
\newcommand{\mE}{\ensuremath{\mat{E}}\xspace}

\newcommand{\mG}{\ensuremath{\mat{G}}\xspace}

\newcommand{\mI}{\ensuremath{\mat{I}}\xspace}

\newcommand{\mL}{\ensuremath{\mat{L}}\xspace}
\newcommand{\mM}{\ensuremath{\mat{M}}\xspace}
\newcommand{\mN}{\ensuremath{\mat{N}}\xspace}

\newcommand{\mQ}{\ensuremath{\mat{Q}}\xspace}
\newcommand{\mR}{\ensuremath{\mat{R}}\xspace}
\newcommand{\mS}{\ensuremath{\mat{S}}\xspace}
\newcommand{\mT}{\ensuremath{\mat{T}}\xspace}
\newcommand{\mU}{\ensuremath{\mat{U}}\xspace}
\newcommand{\mV}{\ensuremath{\mat{V}}\xspace}
\newcommand{\mW}{\ensuremath{\mat{W}}\xspace}
\newcommand{\mX}{\ensuremath{\mat{X}}\xspace}
\newcommand{\mY}{\ensuremath{\mat{Y}}\xspace}
\newcommand{\mZ}{\ensuremath{\mat{Z}}\xspace}
\newcommand{\mLambda}{\ensuremath{\mat{\Lambda}}\xspace}
\newcommand{\mSigma}{\ensuremath{\mat{\Sigma}}\xspace}

\newcommand{\va}{\ensuremath{\vec{a}}\xspace}

\newcommand{\ve}{\ensuremath{\vec{e}}\xspace}

\newcommand{\vg}{\ensuremath{\vec{g}}\xspace}

\newcommand{\vr}{\ensuremath{\vec{r}}\xspace}
\newcommand{\vs}{\ensuremath{\vec{s}}\xspace}
\newcommand{\vt}{\ensuremath{\vec{t}}\xspace}

\newcommand{\vv}{\ensuremath{\vec{v}}\xspace}

\newcommand{\vx}{\ensuremath{\vec{x}}\xspace}

\newcommand{\vz}{\ensuremath{\vec{z}}\xspace}

\newcommand{\vlambda}{\ensuremath{\vecalt{\lambda}}\xspace}

\sodef\allcapsspacing{\upshape}{0.15em}{0.65em}{0.6em}%

\colorlet{todo_background_normal}{white}
\definecolor{todo_background_dark}{RGB}{39,40,34}

\definecolor{advice_text}{RGB}{78, 12, 123}
\colorlet{advice_background}{todo_background_normal}

\definecolor{incomplete_text}{RGB}{204, 64, 84}
\colorlet{incomplete_background}{todo_background_normal}

\newcounter{question}
\newcommand{\R}{\bbR}

\newcommand{\bs}[1]{\mat{#1}}

\newcommand{\norm}[1]{\normof{#1}}

\newcommand{\tr}{\trace}

\DeclareMathOperator{\hutch}{H}

\newcommand{\hutchpp}{\text{Hutch\raisebox{0.35ex}{\relscale{0.75}++}}\xspace} 
\newcommand{\hutchppgauss}{\text{Gaussian-Hutch\raisebox{0.35ex}{\relscale{0.75}++}}\xspace} 
\newcommand{\hutchppbold}{\textbf{Hutch\raisebox{0.35ex}{\relscale{0.75}++}}\xspace}
\newcommand{\hutchppna}{\text{NA-Hutch\raisebox{0.35ex}{\relscale{0.75}++}}\xspace} 
\newcommand{\hutchppnabold}{\textbf{NA-Hutch\raisebox{0.35ex}{\relscale{0.75}++}}\xspace}

\newcommand{\tsfrac}[2]{{\textstyle\frac{#1}{#2}}}

  \usepackage{nth}
  \usepackage{intcalc}

\title{Hutch++: Optimal Stochastic Trace Estimation}

\author{
	Raphael A. Meyer \\ New York University\\ \texttt{ram900@nyu.edu}
	\and 
	Cameron Musco\\ UMass Amherst\\ \texttt{cmusco@cs.umass.edu}
	\and
	Christopher Musco\\ New York University\\ \texttt{cmusco@nyu.edu}
	\and
	David P. Woodruff\\ Carnegie Mellon University\\ \texttt{dwoodruf@cs.cmu.edu}
}

\begin{document}

\maketitle


\begin{abstract}
We study the problem of estimating the trace of a matrix \mA that can only be accessed through matrix-vector multiplication.
We introduce a new randomized algorithm, \hutchpp,  which computes a \((1 \pm \epsilon)\) approximation to \(\tr(\mA)\) for any positive semidefinite (PSD) \mA using just \(O(1/\epsilon)\) matrix-vector products.
This improves on the ubiquitous \emph{Hutchinson's estimator}, which requires \(O(1/\epsilon^2)\) matrix-vector products.
Our approach is based on a simple technique for reducing the variance of Hutchinson's estimator using a low-rank approximation step, and is easy to implement and analyze.
Moreover, we prove that, up to a logarithmic factor, the complexity of \hutchpp is optimal amongst all matrix-vector query algorithms, even when queries can be chosen adaptively.
We show that it significantly outperforms Hutchinson's method in experiments.
While our theory mainly requires \mA to be positive semidefinite, we provide generalized guarantees for general square matrices, and show empirical gains in such applications.
\end{abstract}


\section{Introduction}
\label{sec:intro}
A ubiquitous problem in numerical linear algebra is that of approximating the trace of a \(d\times d\) matrix \mA that can only be accessed via matrix-vector multiplication queries. In other words, we are given access to an oracle that can evaluate \(\mA\vx\) for any \(\vx\in \R^d\), and the goal is to return an approximation to $\tr(\mA)$ using as few queries to this oracle as possible. An exact solution can be obtained with $d$ queries because $\tr(\mA) = \sum_{i=1}^d \mA_{ii} = \sum_{i=1}^d \ve_i^T \mA\ve_i$, where $\ve_i$ denotes the $i^\text{th}$ standard basis vector. The goal is thus to develop algorithms that use far fewer than $d$ matrix-vector multiplications. 

Known as \emph{implicit} or \emph{matrix free} trace estimation, this problem arises in applications that require the trace of a matrix $\mA$, where $\mA$ is itself a transformation of some other  matrix $\mB$. For example, $\mA = \mB^q$, $\mA = \mB^{-1}$, or $\mA = \exp(\mB)$. In all of these cases, explicitly computing $\mA$  would require roughly $O(d^3)$ time, whereas multiplication with a vector $\vx$ can be implemented more quickly using iterative methods. For example, $\mB^q\vx$ can be computed in just $O(d^2)$ time for constant $q$, and for well-conditioned matrices, $\mB^{-1}\vx$ and $\exp(\mB)\vx$ can also be computed in $O(d^2)$ time using the conjugate gradient or Lanczos methods \cite{Higham:2008}. 
Implicit trace estimation is used to approximate matrix norms \cite{HanMalioutovAvron:2017,MuscoNetrapalliSidford:2018},  spectral densities \cite{LinSaadYang:2016,Cohen-SteinerKongSohler:2018,BravermanKrauthgamerKrishnan:2020}, log-determinants \cite{BoutsidisDrineasKambadur:2015,HanMalioutovShin:2015}, the Estrada index \cite{UbaruSaad:2018,WangSunMusco:2020}, eigenvalue counts in intervals \cite{Di-NapoliPolizziSaad:2016}, triangle counts in graphs \cite{Avron:2010}, and much more \cite{Chen:2016}. In these applications, we typically have that \mA is symmetric, and often positive semidefinite (PSD).

\subsection{Hutchinson's Estimator}
The most common method for implicit trace estimation is {Hutchinson's stochastic estimator} \cite{Hutchinson:1990}. This elegant randomized algorithm works as follows: let $\mG = [\vg_1, \ldots, \vg_m]\in \R^{d\times m}$ be a matrix containing i.i.d. random variables with mean $0$ and variance $1$. A simple calculation shows that $\E [\vg_i^T \mA\vg_i] = \tr(\mA)$ for each $\vg_i\in \R^d$, and $\vg_i^T \mA\vg_i$ can be computed with just one matrix-vector multiplication. 
So to approximate $\tr(\mA)$, Hutchinson's estimator returns the following average:
\begin{align}
\label{eq:hutch}
\text{Hutchinson's Estimator:} &&\hutch_m(\mA) = \frac{1}{m} \sum_{i=1}^m \vg_i^T\mA\vg_i = \frac{1}{m}\tr(\mG^T\mA\mG).
\end{align} 
Hutchinson's original work suggests using random $\pm 1$ sign vectors for $\vg_1, \ldots, \vg_m$, and an earlier paper by Girard suggests standard normal random variables \cite{Girard:1987}. Both choices perform similarly, as both random variables are sub-Gaussian. For vectors with sub-Gaussian random entries, it can be proven that, when $\mA$ is positive semidefinite, $(1-\epsilon)\tr(\mA) \leq \hutch_m(\mA) \leq (1+\epsilon) \tr(\mA)$ with probability $\ge 1-\delta$ if we use $m=O \left ( {\log(\nicefrac1\delta)/\epsilon^2}\right )$ matrix-vector multiplication queries \cite{AvronToledo:2011,Roosta-KhorasaniAscher:2015}.\footnote{For non-PSD matrices, this generalizes to $\tr(\mA) - \epsilon \|\mA\|_F \leq \hutch_m(\mA) \leq \tr(\mA) + \epsilon \|\mA\|_F$, which implies the relative error bound since when $\mA$ is PSD, $\|\mA\|_F \le \tr(\mA)$.} For constant $\delta$ (e.g., $\delta = \nicefrac1{10}$) the bound is $O(1/\epsilon^2)$.

\subsection{Our results}
Since Hutchinson's work, and the non-asymptotic analysis in \cite{AvronToledo:2011}, there has been no improvement on this $O(1/\epsilon^2)$ matrix-vector multiplication bound for trace approximation. 
Our main contribution is a quadratic improvement: we provide a new algorithm, \hutchpp, that obtains the same $(1\pm\epsilon)$ guarantee with $O(1/\epsilon)$ matrix-vector multiplication queries.
This algorithm is nearly as simple as the original Hutchinson's method, and can be implemented in just a few lines of code. 

\begin{algorithm}[H]
	\caption{\hutchpp}
	\label{alg:hutchpp-qr}
	{\bfseries input}: Matrix-vector multiplication oracle for matrix, \(\mA\in\bbR^{d \times d}\). Number of queries, \(m\). \\
	{\bfseries output}: Approximation to \(\tr(\mA)\).\\
	\vspace{-1em}
	\begin{algorithmic}[1]
		\STATE Sample \(\mS\in\bbR^{d \times \frac m3}\) and \(\mG\in\bbR^{d \times \frac m3}\) with \iid $\{+1,-1\}$ entries.
		\STATE Compute an orthonormal basis $\mQ \in \bbR^{d \times \frac m3}$ for the span of $\mA\mS$ (e.g., via QR decomposition).
		\STATE {\bfseries return} \(\hutchpp(\mA) = \trace(\mQ^T\mA\mQ) + \frac 3m \trace(\mG^T(\eye-\mQ\mQ^T)\mA(\eye-\mQ\mQ^T)\mG)\).
	\end{algorithmic}
\end{algorithm}
\noindent\hutchpp requires $m$ matrix-vector multiplications with $\mA$: $m/3$ to compute $\mA\cdot\mS$, $m/3$ to compute $\mA\cdot\mQ$, and $m/3$ to compute $\mA\cdot (\eye - \mQ\mQ^T)\mG$. It requires $O(dm^2)$ additional runtime to compute the basis $\mQ$ and the product $(\eye-\mQ\mQ^T)\mG = \mG -\mQ\mQ^T\mG$.
For concreteness, we state the method with random sign matrices, but the entries of $\mS$ and $\mG$ can be any sub-Gaussian random variables with mean $0$ and variance $1$, including e.g., standard Gaussians. Our main theorem on \hutchpp is:

\begin{theorem}
	\label{thm:intro_thm}
	If \hutchpp is implemented with $m = O({\sqrt{\log(\nicefrac1\delta)}/\eps} + \log(\nicefrac1\delta))$ matrix-vector multiplication queries, then for any PSD $\mA$, with probability $\ge 1-\delta$, the output $\hutchpp(\mA)$ satisfies:
	\begin{align*}
	(1-\epsilon)\tr(\mA) \leq \hutchpp(\mA) \leq (1+\epsilon)\tr(\mA).
	\end{align*}
\end{theorem}
\hutchpp can be viewed as a natural \emph{variance reduced} version of Hutchinson's estimator. The method starts by computing an orthonormal span $\mQ \in \R^{d\times \frac m3}$ by running a single iteration of power method with a random start matrix $\mS$.  $\mQ$ coarsely approximates the span of $\mA$'s top eigenvectors.  Then we separate $\mA$ into its projection onto the subspace spanned by $\mQ$, and onto that subspace's orthogonal compliment, writing $\tr(\mA) = \tr(\mQ\mQ^T\mA\mQ\mQ^T) + \tr\left((\eye - \mQ\mQ^T)\mA(\eye - \mQ\mQ^T)\right)$. By the cyclic property of the trace, the first term is equal to $\tr(\mQ^T\mA\mQ)$, which is computed \emph{exactly} by \hutchpp with $\frac m3$ matrix-vector multiplications. The second term is approximated using Hutchinson's estimator with the random vectors in $\mG$. 

Thus, the error in estimating $\tr(\mA)$ is entirely due to approximating this second term. The key observation is that the variance when estimating this term is much lower than when estimating $\tr(\mA)$ directly. Specifically, it is proportional to $\|(\eye - \mQ\mQ^T)\mA(\eye - \mQ\mQ^T)\|_F^2$, which, using 
%
standard tools from randomized linear algebra \cite{CohenElderMusco:2015,Woodruff:2014}, we can show is bounded by $\epsilon \tr(\mA)^2$ with good probability when $m = O(1/\epsilon)$.
This yields our improvement over Hutchinson's method applied directly to $\mA$, which has variance bounded by  $\tr(\mA)^2$.
%
The full proof of \theoremref{intro_thm} is in  \sectionref{upper_bound_proof}.


\algorithmref{hutchpp-qr} is \emph{adaptive}: it multiplies $\mA$ by a sequence of query vectors $\vr_1, \ldots, \vr_m$, where later queries depend on earlier ones. In contrast, Hutchinson's method is \emph{non-adaptive}: $\vr_1, \ldots, \vr_m$ are chosen in advance, before computing any of the products $\mA\vr_1, \ldots, \mA\vr_m$. In addition to \algorithmref{hutchpp-qr}, we give a non-adaptive variant of \hutchpp that obtains the same $ O(1/\epsilon)$ bound. 
We complement these results with a nearly matching lower bound, proven in \sectionref{lower_bound}. Specifically, via a reduction from the {Gap-Hamming problem} from communication complexity, we show that any matrix-vector query algorithm whose queries have bounded bit complexity requires $m = \Omega\left(\frac{1}{\epsilon \log(1/\epsilon)}\right)$ queries to estimate the trace of a PSD matrix up to a $(1\pm \epsilon)$ multiplicative approximation. We also prove a tight $m = \Omega\left(\frac{1}{\epsilon}\right)$ lower bound for non-adaptive algorithms in the real RAM model of computation.

Finally, we provide a generalization of our upper bound to non-PSD matrices. We also provide an analysis of the variance the \hutchpp estimator, which compliments the high-probability bound of  \theoremref{intro_thm}. The variance analysis involves explicit constants, which may be useful to practioners.

\medskip\noindent\textbf{Empirical Results.}
In \sectionref{experiments} we present experimental results on synthetic and real-world matrices, including applications of trace estimation to approximating log determinants, the graph Estrada index, and the number of triangles in a graph.
We demonstrate that \hutchpp improves substantially on Hutchinson's estimator, and on related estimators based on approximating the top eigenvalues of $\mA$.
While our relative error bounds only apply to PSD matrices, \hutchpp can be used unmodified on general square matrices, and we experimentally confirm that it still outperforms Hutchinson's in this case.
We note that \hutchpp is simple to implement and essentially parameter free: the only choice needed is the number of matrix-vector multiplication queries $m$. 

\subsection{Prior Work}
\textbf{Upper bounds.}
A nearly tight non-asymptotic analysis of Hutchinson's estimator for positive semidefinite matrices was given by Avron and Toledo using an approach based on reducing to Johnson-Lindenstrauss random projection \cite{AvronToledo:2011,DasguptaGupta:2003,Achlioptas:2003}. A slightly tighter approach from \cite{Roosta-KhorasaniAscher:2015} obtains a $(1\pm\epsilon)$ multiplicative error bound with $m= O(1/\epsilon^2)$ matrix-vector multiplication queries. This bound is what we improve on with \hutchpp. A more in depth discussion of different variations on Hutchinson's method and existing error bounds can be found in survey of \cite{MartinssonTropp:2020}. 

A number of papers suggest variance reduction schemes for Hutchinson's estimator. Some take advantage of sparsity structure in $\mA$ \cite{TangSaad:2011,StathopoulosLaeuchliOrginos:2013} and others use a ``decomposition'' approach similar to \hutchpp
\cite{AdamsPenningtonJohnson:2018}. 
Most related to our work are two papers which, like \hutchpp, perform the decomposition by projecting onto some $\mQ$ that approximately spans $\mA$'s top eigenspace \cite{GambhirStathopoulosOrginos:2017,Lin:2017}. The justification is that this method should perform much better than Hutchinson's when $\mA$ is close to low-rank, because $\tr(\mQ^T\mA\mQ)$ will capture most of $\mA$'s trace.
Our contribution is an analysis of this approach which 1) improves on Hutchinson's \emph{even when $\mA$ is far from low-rank} and 2) shows that a very coarse approximation to the top eigenvectors suffices (computed using one iteration of the power method). 
%
%
Finally, we note two papers which directly use the approximation $\tr(\mA) \approx \tr(\mQ^T\mA\mQ)$, where $\mQ$ is computed with a  randomized SVD method \cite{SaibabaAlexanderianIpsen:2017,Hanyu-Li:2020}. Of course, this approach works best for nearly-low rank matrices. 

\vspace{.5em}
\medskip\noindent\textbf{Lower bounds.}
Our lower bounds extend a recent line of work on lower bounds for linear algebra problems in the ``matrix-vector query model'' \cite{SimchowitzEl-AlaouiRecht:2018, SunWoodruffYang:2019,BravermanHazanSimchowitz:2020}. 
 \cite{WimmerWuZhang:2014} proves a lower bound of $\Omega(1/\epsilon^2)$ queries for PSD trace approximation in an alternative model that allows for adaptive ``quadratic form'' queries: $\vr_1^T \mA \vr_1, \ldots, \vr_m^T \mA \vr_m$. This model captures Hutchinson's estimator, but not \hutchpp, which is why we are able to obtain an upper bound of $O(1/\epsilon)$ queries.




\section{Preliminaries}
\textbf{Notation.} 
For $\va \in \R^d$, $\|\va\|_2 = (\sum_{i=1}^d a_i^2)^{1/2}$ denotes the $\ell_2$ norm and $\|\va\|_1 = \sum_{i=1}^d |a_i|$ denotes the $\ell_1$ norm. For $\mA\in \R^{n\times d}$, $\|\mA\|_F = (\sum_{i=1}^n\sum_{j=1}^d \mA_{ij}^2)^{1/2}$ denotes the Frobenius norm. For square $\mA \in \R^{d\times d}$, $\tr(\mA) = \sum_{i=1}^d \mA_{ii}$ denotes the trace. 
Our main results on trace approximation are proven for symmetric positive semidefinite (PSD) matrices, which are the focus of many applications. Any symmetric $\mA\in\R^{d\times d}$ has eigendecomposition $\mA = \mV\mLambda \mV^T$, where $\mV \in \R^{d\times d}$ is orthogonal and $\mLambda$ is a real-valued diagonal matrix.
We let $\vlambda = \diag(\mLambda)$ be a vector containing $\mA$'s eigenvalues in descending order: $\lambda_1 \geq \lambda_2 \geq \ldots \geq \lambda_d$. 
When $\mA$ is PSD, $\lambda_i \geq 0$ for all $i$. 
We use the identities $\tr(\mA) = \norm{\vlambda}_1$ and $\|\mA\|_F = \norm{\vlambda}_2$. We let $\mA_k = \arg\min_{\mB, \rank(\mB) = k} \|\mA - \mB\|_F$ denote the optimal $k$-rank approximation to $\mA$. For a PSD matrix $\mA$, $\mA_k = \mV_k\mLambda_k \mV_k^T$, where $\mV_k\in \R^{d\times k}$ contains the first $k$ columns of $\mV$ and $\mLambda_k$ is the $k\times k$ top left submatrix of $\mLambda$. 

We state a few results for non-PSD matrices which depend on the nuclear norm. Consider a general square matrix $\mA\in\R^{d\times d}$ with singular value decomposition $\mA = \mU\mSigma \mV^T$, where $\mV,\mU \in \R^{d\times d}$ is orthogonal and $\mSigma$ is a positive diagonal matrix containing $\mA$'s singular values, $\sigma_1, \ldots, \sigma_d$. The nuclear norm $\|\mA\|_*$ is equal to $\|\mA\|_* = \sum_{i=1}^d \sigma_i$. For PSD $A$, $\|\mA\|_* = \tr(\mA)$. 

\medskip \noindent\textbf{Hutchinson's Analysis.} We require a standard bound on the accuracy of Hutchinson's estimator:
\begin{lemma}
	\label{lem:hutch-frob}
	Let \(\mA\in\bbR^{d \times d}\), \(\delta\in(0, \nicefrac12]\), $\ell\in\bbN$. Let  $\hutch_\ell(\mA)$ be the  \(\ell\)-query Hutchinson estimator defined in \eqref{eq:hutch}, implemented with mean 0, i.i.d. sub-Gaussian random variables with constant sub-Gaussian parameter. For fixed constants $c,C$, if $\ell > c\log(\nicefrac1\delta)$, then with probability \( \ge 1-\delta\), 
\[
\abslr{\hutch_\ell(\mA) - \trace(\mA)} \leq C\sqrt{\frac{\log(\nicefrac1\delta)}{\ell}} \normof\mA_F.
\]
So, if $\ell = O\left(\frac{\log(\nicefrac1\delta)}{\epsilon^2}\right)$ then, 
with probability \(\ge 1-\delta\),
$\abslr{\hutch_\ell(\mA) - \trace(\mA)} \leq \epsilon \normof\mA_F$.
\end{lemma}
We refer the reader to \cite{rudelson2013hanson} for a formal definition of sub-Gaussian random variables: both normal $\cN(0,1)$ random variables and $\pm 1$ random variables are sub-Gaussian with constant parameter. \lemmaref{hutch-frob} is proven in  \appendixref{hutch-frob} for completeness. It is slightly more general than prior work \cite{Roosta-KhorasaniAscher:2015} in that it applies to non-PSD, and even asymmetric matrices, which will be important in the analysis of our non-adaptive algorithm. A similar result was recently shown in \cite{CortinovisKressner:2020}. 


\section{Complexity Analysis}
\label{sec:upper_bound_proof}

We start by providing the technical intuition behind \hutchpp. First note that, for a PSD matrix with eigenvalues $\vlambda$, $\|\mA\|_F  \leq  \tr(\mA)$, so \lemmaref{hutch-frob} immediately implies that Hutchinson's estimator obtains a relative error guarantee with $O(1/{\eps^2})$ queries. However, this bound is only tight when $\|\vlambda\|_2 \approx \|\vlambda\|_1$, i.e., when $\mA$ has \emph{significant mass concentrated on just a small number of eigenvalues.} 

\hutchpp simply eliminates this possibility by approximately projecting off \mA's large eigenvalues using a projection $\mQ\mQ^T$. By doing so, it only needs to compute a stochastic estimate for the trace of $(\eye - \mQ\mQ^T)\mA(\eye - \mQ\mQ^T)$. The error of this estimate is proportional to $\|(\eye - \mQ\mQ^T)\mA(\eye - \mQ\mQ^T)\|_F$, which we show is \emph{always much smaller than $\tr(\mA)$}. 
In particular, suppose that $\mQ = \mV_k$ exactly spanned the top $k$ eigenvectors \mA and thus $(\eye - \mQ\mQ^T)\mA(\eye - \mQ\mQ^T) = \mA - \mA_k$. Then we have:

\begin{lemma}
\label{lem:low-rank-frob-trace}
For any PSD matrix \mA,
\(
	\normof{\mA-\mA_k}_F \leq \frac1{\sqrt k} \trace(\mA).
\)
\end{lemma}
\begin{proof}
We have \(\lambda_{k+1} \leq \frac1k \sum_{i=1}^k \lambda_i \leq \frac1k \trace(\mA)\), so:
\[
	\normof{\mA-\mA_k}_F^2
	= \sum_{i=k+1}^d \lambda_i^2
	\leq \lambda_{k+1} \sum_{i=k+1}^d \lambda_i
	\leq \frac1k \trace(\mA) \sum_{i=k+1}^d \lambda_i
	\leq \frac1k \trace(\mA)^2.
	\qedhere
\]
\end{proof}
The above analysis can be tightened by a factor of two via Lemma 7 in \cite{DBLP:conf/stoc/GilbertSTV07}.

\lemmaref{low-rank-frob-trace} immediately suggests the possibility of an algorithm with $O(1/\epsilon)$ query complexity: Set $k = O(1/\epsilon)$ and split $\tr(\mA) = \tr(\mA_k) + \tr(\mA - \mA_k)$. The first term can be computed exactly with $O(1/\epsilon)$ matrix-vector multiplication queries if $\mV_k$ is known, since $\tr(\mA_k)= \tr(\mV_k^T\mA\mV_k)$. By \lemmaref{low-rank-frob-trace} combined with \lemmaref{hutch-frob}, the second can be estimated to error $\pm\epsilon \tr(\mA)$ using just $O(1/\epsilon)$ queries instead of $O(1/\epsilon^2)$. 
Of course, we can't compute $\mV_k$ exactly with a small number of matrix-vector multiplication queries, but this is easily resolved by using an approximate projection. Using standard tools from randomized linear algebra, $O(k)$ queries suffices to find a $\mQ$ with $\norm{(\eye - \mQ\mQ^T)\mA(\eye - \mQ\mQ^T)}_F \leq O(\norm{\mA-\mA_k}_F)$, which is all that is needed for a $O(1/\epsilon)$ query result.

Concretely, we use \lemmaref{low-rank-frob-trace} to prove the following general theorem, from which \theoremref{intro_thm} and our non-adaptive algorithmic result will follow as direct corollaries.

\begin{theorem}
\label{thm:hutchpp}
Let \(\mA\in\bbR^{d \times d}\) be PSD, \(\delta\in(0,\nicefrac12)\), \(\ell\in\bbN\), \(k\in\bbN\). Let $\tilde \mA$ and $\bs{\Delta}$ be any matrices with:
\begin{align*}
\tr(\mA) &= \tr(\tilde \mA) + \tr(\bs{\Delta}) & &\text{and} & \|\bs{\Delta}\|_F \leq 2\|\mA - \mA_k\|_F.
\end{align*}
For fixed constants $c,C$, if $\ell > c\log(\nicefrac1\delta)$, then with probability \(1-\delta\), $Z = \left[\tr(\tilde \mA) +\hutch_\ell(\bs{\Delta})\right]$ satisfies:
\begin{align*}
	\Bigl\lvert Z - \trace(\mA)\Bigr\rvert \leq 2C\sqrt{\tsfrac{\log(\nicefrac1\delta)}{k\ell}} \cdot \trace(\mA).
\end{align*}
In particular, if \(k = \ell = O\left(\tsfrac{\sqrt{\log(\nicefrac1\delta)}}{\eps} + \log(\nicefrac1\delta)\right)\), $Z$  is a \((1\pm \epsilon)\) error approximation to \(\tr(\mA)\).
\end{theorem}
\begin{proof}
We have with probability $\ge 1-\delta$:
\begin{align*}
\abs{Z - \trace(\mA)}
&= \abs{\hutch_\ell(\bs{\Delta}) - \trace(\bs{\Delta})}\tag{since $Z = \tr(\tilde \mA) +\hutch_\ell(\bs{\Delta})$ and $\tr(\mA) = \tr(\tilde \mA) + \tr(\bs{\Delta}) $} \\
&\leq  C\sqrt{\tsfrac{\log(\nicefrac1\delta)}{\ell}} \normof{\bs{\Delta}}_F \tag{by the standard Hutchinson's analysis, \lemmaref{hutch-frob}}\\
& \leq 2C\sqrt{\tsfrac{\log(\nicefrac1\delta)}{\ell}}\normof{\mA-\mA_k}_F\tag{by the assumption that $\|\bs{\Delta}\|_F \le  2 \|\mA -\mA_k\|_F$} \\
&\leq 2C\sqrt{\tsfrac{\log(\nicefrac1\delta)}{k\ell}} \tr(\mA).\tag{by \lemmaref{low-rank-frob-trace}}
\end{align*}
\end{proof}

As discussed, \theoremref{hutchpp} would immediately yield an $O(1/\epsilon)$ query algorithm if we knew an optimal $k$-rank approximation for $\mA$. Since computing one is infeasible, our first version of \hutchpp (\algorithmref{hutchpp-qr}) instead uses a projection onto a subspace $\mQ$ which is computed with one iteration of the power method. We have:


\begin{reptheorem}{intro_thm}
	If \algorithmref{hutchpp-qr} is implemented with $m = O({\sqrt{\log(\nicefrac1\delta)}/\eps} + \log(\nicefrac1\delta))$ matrix-vector multiplication queries, then for any PSD $\mA$, with probability $\ge 1-\delta$, the output $\hutchpp(\mA)$ satisfies:
	$(1-\epsilon)\tr(\mA) \leq \hutchpp(\mA) \leq (1+\epsilon)\tr(\mA)$.
\end{reptheorem}
\begin{proof}
	Let \mS, \mG, and \mQ be as in \algorithmref{hutchpp-qr}. We instantiate \theoremref{hutchpp} with $\tilde \mA = \mQ^T\mA\mQ$ and $\bs{\Delta} = (\eye - \mQ\mQ^T)\mA(\eye - \mQ\mQ^T)$. Note that, since $\mQ$ is orthogonal, $(\eye - \mQ\mQ^T)$ is a projection matrix, so $(\eye - \mQ\mQ^T) = (\eye - \mQ\mQ^T)^2$. This fact, along with the cyclic property of the trace, gives:
	\begin{align*}
	\tr(\tilde \mA) &= \tr(\mA\mQ\mQ^T) & &\text{and} & \tr(\bs{\Delta}) &= \tr(\mA(\eye - \mQ\mQ^T)),
	\end{align*}
	and thus $\tr(\tilde\mA) + \tr(\bs{\Delta}) = \tr(\mA)$ as required by \theoremref{hutchpp}. Furthermore, since multiplying by a projection matrix can only decrease Frobenius norm, $\|\bs{\Delta}\|_F^2 \leq \|\mA(\eye - \mQ\mQ^T)\|_F^2 = \|\mA - \mA\mQ\mQ\|_F^2$. 
	
	Recall that $\mQ$ is an orthogonal basis for the column span of $\mA\mS$, where $\mS$ is a random sign matrix with $\frac m3$ columns. $\mQ$ is thus an orthogonal basis for a linear sketch of $\mA$'s column space, and it is well known that $\mQ$ will align with large eigenvectors of $\mA$, and $\|\mA - \mA\mQ\mQ^T\|_F^2$ will be small \cite{Sarlos:2006,Woodruff:2014}. Concretely, applying Corollary 7 and Claim 1 from \cite{DBLP:journals/corr/abs-2004-08434}, we have that, as long as $\frac m3 \geq O(k + \log(\nicefrac1\delta))$, with probability $\ge 1-\delta$:
	\begin{align*}
	\|\mA - \mA\mQ\mQ^T\|_F^2 \leq 2 \|\mA - \mA_k\|_F^2.
	\end{align*}
	Accordingly, $\|\bs{\Delta}\|_F \leq 2 \|\mA - \mA_k\|_F^2$ as required by \theoremref{hutchpp}. The result then immediately follows by setting $k = O({\sqrt{\log(\nicefrac1\delta)}/\eps} + \log(\nicefrac1\delta))$ and noting that $\hutchpp(\mA) = \left[\tr(\tilde \mA) +\hutch_\ell(\bs{\Delta})\right]$ where $\ell = O({\sqrt{\log(\nicefrac1\delta)}/\eps} + \log(\nicefrac1\delta))$.
\end{proof}

Notably, none of the analysis above uses the fact that \mA is PSD except for \lemmaref{low-rank-frob-trace}.
However, \lemmaref{low-rank-frob-trace} holds for any matrix by replacing the trace with the nuclear norm (the two are equal for PSD matrices).
So, the following result holds for general square matrices:
\begin{theorem}
	\label{thm:general-hutchpp-qr}
	If \algorithmref{hutchpp-qr} is implemented with \(m = O({\sqrt{\log(\nicefrac1\delta)}/\eps} + \log(\nicefrac1\delta))\) matrix-vector multiplication queries, then for any \mA, with probability \(\ge 1-\delta\), the output \(\hutchpp(\mA)\) satisfies:
	\[
		\abs{\hutchpp(\mA) - \trace(\mA)}
		\leq \sqrt\eps \normof{\mA-\mA_{1/\eps}}_F
		\leq \eps \normof{\mA}_*.
	\]
\end{theorem}
Using the same number of queries, Hutchinson's estimator achieves a bound of \(O(\sqrt\eps) \cdot \normof{\mA}_F\).
Since \(\normof{\mA-\mA_{1/\eps}}_F \leq \normof\mA_F\), the first inequality in \theoremref{general-hutchpp-qr} shows that \hutchpp is never asymptotically slower than Hutchinson's, even for non-PSD matrices.
Furthermore, if \mA has quickly decaying eigenvalues, this inequality shows that \hutchpp will converge especially quickly.
The second inequality mirrors \theoremref{intro_thm}, stating the deviation of \hutchpp in terms of nuclear norm instead of trace. For PSD matrices, the two are equivalent.


\subsection{A Non-Adaptive Variant of \texorpdfstring{\hutchpp}{Hutch++}}
As discussed in \sectionref{intro},
\algorithmref{hutchpp-qr} is \emph{adaptive}: it uses the result of computing \(\mA\mS\) to compute \(\mQ\), which is then multiplied by \mA to compute the \(\trace(\mQ^T\mA\mQ)\) term.
Meanwhile, Hutchinson's estimator is \emph{non-adaptive}: it samples a single random matrix upfront, batch-multiplies by \mA once, and computes an approximation to $\tr(\mA)$ from the result, without any further queries.

Not only is non-adaptivity an interesting theoretical property, but it can be practically useful, since parallelism or block iterative methods often make it faster to multiply an implicit matrix by many vectors at once. With these considerations in mind, we describe a non-adaptive variant of \hutchpp, which we call \hutchppna. \hutchppna obtains nearly the same theoretical guarantees as \algorithmref{hutchpp-qr}, although it tends to perform slightly worse in our experiments.
%
%

We leverage a streaming low-rank approximation result of Clarkson and Woodruff \cite{DBLP:conf/stoc/ClarksonW09} which shows that  if \(\mS\in\bbR^{d \times m}\) and \(\mR\in\bbR^{d \times cm}\) are sub-Gaussian random matrices with \(m = O(k\log(\nicefrac1\delta))\) and $c > 1$ a fixed constant, then with probability \(1-\delta\), the matrix \(\tilde\mA = \mA\mR(\mS^T\mA\mR)^+(\mA\mS)^T\) satisfies \(\normof{\mA-\tilde\mA}_F \leq 2 \normof{\mA-\mA_k}_F\). Here $^+$ denotes the Moore-Penrose pseudoinverse.
We can compute \(\trace(\tilde\mA)\) efficiently without explicitly constructing \(\tilde\mA\in \R^{d\times d}\)  by noting that it is equal to \(\trace((\mS^T\mA\mR)^+(\mA\mS)^T(\mA\mR))\) via the cyclic property of the trace.
This yields:
\begin{algorithm}[H]
	\caption{\hutchppna (Non-Adaptive variant of \hutchpp)}
	\label{alg:hutchpp-streaming}
	{\bfseries input}: Matrix-vector multiplication oracle for matrix, \(\mA\in\bbR^{d \times d}\). Number of queries, \(m\). \\
	{\bfseries output}: Approximation to \(\tr(\mA)\).\\
	\vspace{-1em}
	%
\begin{algorithmic}[1]
	\STATE Fix constants $c_1,c_2,c_3$ such that $c_1 < c_2$ and $c_1 + c_2 + c_3 = 1$. 
	\STATE Sample \(\mS\in\bbR^{d \times c_1 m}\), \(\mR\in\bbR^{d \times c_2 m}\), and \(\mG\in\bbR^{d \times  c_3 m}\)  with \iid $\{+1,-1\}$ entries.
	\STATE Compute \(\mZ = \mA\mR\) and \(\mW = \mA\mS\).
	\STATE {\bfseries return} \(\hutchppna(\mA)  = \trace( (\mS^T\mZ)^+(\mW^T\mZ)) + \frac{1}{c_3m} \left [\trace(\mG^T\mA\mG) - \trace(\mG^T\mZ(\mS^T\mZ)^+\mW^T\mG) \right ]\)
\end{algorithmic}
\end{algorithm}
\noindent \hutchppna requires $m$ matrix-vector multiplications with $\mA$. In our experiments, it works well with $c_1 = c_3 = \nicefrac{1}{4}$ and $c_2 = \nicefrac{1}{2}$. Assuming $m < d$, it requires $O(dm^2)$ further runtime, to perform the matrix multiplications on line $4$ and to compute $(\mS^T \mZ)^+$,  which takes $O(dm^2 + m^3)$ time.
 
\begin{theorem}\label{thm:hutchna}
	If \hutchppna is implemented with $m = O(\log(\nicefrac1\delta)/\eps)$ matrix-vector multiplication queries and $\frac{c_2}{c_1}$ a sufficiently large constant, then for any PSD $\mA$, with probability $\ge 1-\delta$, the output $\hutchppna(\mA)$ satisfies:
	$(1-\epsilon)\tr(\mA) \leq \hutchppna(\mA) \leq (1+\epsilon)\tr(\mA)$.
%
\end{theorem}
\begin{proof}
We apply \theoremref{hutchpp} with $\tilde\mA = \mZ(\mS^T\mZ)^+\mW^T$, $\bs{\Delta} = \mA - \tilde\mA$, $k = O(1/\epsilon)$ and $\ell = c_3 m = O(\frac{\log(\nicefrac1\delta)}{\epsilon})$. $\tr(\mA) = \tr(\tilde \mA) + \tr(\bs{\Delta})$ and $\hutchppna(\mA) = [\tr(\tilde \mA) + \hutch_\ell(\bs{\Delta})]$. 
By Theorem 4.7 of \cite{DBLP:conf/stoc/ClarksonW09}, since $c_1 m = O(k \log(\nicefrac1\delta))$,  $\norm{\bs{\Delta}}_F \le 2 \norm{\mA-\mA_k}_F$ with probability $\ge 1-\delta$ as required. 
\end{proof}


\section{Lower Bounds}
\label{sec:lower_bound}
A natural question is if the $O(1/\epsilon)$ matrix-vector query bound of \theoremref{intro_thm} and \theoremref{hutchna} is tight. In this section, we prove that it is up to a logarithmic factor, even for algorithms that perform \emph{adaptive} queries like \hutchpp. Our lower bound is via a reduction to communication complexity: we show that a better algorithm for PSD trace estimation would imply a better 2-party communication protocol for the Gap-Hamming problem, which would violate known adaptive lower bounds for that problem \cite{ChakrabartiRegev:2012}. To prove this result we need to assume a fixed precision model of computation. Specifically we require that the entries in each query vector $\vr$ are integers bounded in absolute value by $2^b$, for some fixed constant $b$. By scaling, this captures the setting where the query vectors are non-integer, but have bounded precision. Formally, we prove in \sectionref{adaptive}:
\begin{theorem}
	\label{thm:adaptive_lower_bound}
	Any algorithm that accesses a positive semidefinite matrix $\mA$ via matrix-vector multiplication queries  $\mA\vr_1, \ldots, \mA\vr_m$, where $\vr_1,\ldots, \vr_m$ are possibly adaptively chosen vectors with integer entries in $\{-2^b,\ldots, 2^b\}$, requires $m = \Omega\left(\frac{1}{\epsilon(b + \log(\nicefrac1\epsilon))}\right)$ such queries to output an estimate $t$ so that, with probability $>\nicefrac 23$, $(1-\epsilon)\tr(\mA)\leq t \leq (1+\epsilon)\tr(\mA)$. 
\end{theorem}
For constant $b$ our lower bound is $\Omega\left(\frac{1}{\epsilon\log(\nicefrac1\epsilon)}\right)$, which matches \theoremref{intro_thm} and \theoremref{hutchna} up to a $\log(\nicefrac1\epsilon)$ factor. 
We also provide an alternative lower bound which holds in the real RAM model of computation (all inputs and arithmetic operations involve real numbers). This second lower bound is tight up to constants, but only applies to non-adaptive algorithms. It is proven using different information theoretic techniques -- we reduce to a hypothesis testing problem involving \emph{negatively spiked} covariance matrices \cite{CaiMaWu:2015,PerryWeinBandeira:2018}. Formally, we prove in \appendixref{lower_bound}:

\begin{theorem}
	\label{thm:nonadaptive_lower_bound}
	Any algorithm that accesses a postive semidefinite matrix $\mA$ through matrix-vector multiplication queries  $\mA\vr_1, \ldots, \mA\vr_m$, where $\vr_1,\ldots, \vr_m$ are real valued \emph{non-adaptively} chosen vectors requires $m = \Omega\left(\frac{1}{\epsilon}\right)$ such queries to output an estimate $t$ so that, with probability $> \nicefrac 34$, $(1-\epsilon)\tr(\mA)\leq t \leq (1+\epsilon)\tr(\mA)$. 
\end{theorem}

%

\subsection{Adaptive lower bound}\label{sec:adaptive}
The proof of \theoremref{adaptive_lower_bound} is based on reducing the Gap-Hamming problem to trace estimation. This problem has been well studied in communication complexity since its introduction in \cite{IndykWoodruff:2003}. 
\begin{problem}[Gap-Hamming]
	\label{prob:gap-hamming}
	Let Alice and Bob be communicating parties who hold vectors $\vs\in \{-1,1\}^n$ and $\vt\in \{-1,1\}^n$, respectively. The Gap-Hamming problem asks Alice and Bob to return:
	\begin{align*}
		1 &\text{ if } \langle \vs,\vt \rangle \geq \sqrt{n} & &\text{and} & -1 &\text{ if } \langle \vs,\vt \rangle \leq -\sqrt{n}. 
	\end{align*}
\end{problem}
A tight lower bound on the unbounded round, randomized communication complexity of this problem was first proven in \cite{ChakrabartiRegev:2012}, with alternative proofs appearing in \cite{Vidick:2012,Sherstov:2012}. Formally:
\begin{lemma}[Theorem 2.6 in \cite{ChakrabartiRegev:2012}]
	\label{lem:gh_lowerbound}
	The randomized communication complexity for solving \problemref{gap-hamming} with probability $\geq \nicefrac 23$ is $\Omega(n)$ bits. 
\end{lemma}
With \lemmaref{gh_lowerbound} in place, we have all we need to prove \theoremref{adaptive_lower_bound}. 
\begin{proof}[Proof of \theoremref{adaptive_lower_bound}\hspace{-0.15cm}]
	Fix a perfect square \(n\in\bbN\).
	Consider an instance of  \problemref{gap-hamming} with inputs $\vs\in \R^{n}$ and $\vt\in \R^{n}$. Let $\mS \in \R^{\sqrt{n}\times\sqrt{n}}$ and $\mT \in \R^{\sqrt{n}\times\sqrt{n}}$ contain the entries of $\vs$ and $\vt$ rearranged into matrices (e.g., placed left-to-right, top-to-bottom). Let $\mZ = \mS + \mT$ and let $\mA = \mZ^T\mZ$. $\mA$ is positive semidefinite and we have:
	\begin{align*}
		\tr(\mA) = \|\mZ\|_F^2 = \|\vs + \vt\|_2^2 &= \|\vs\|_2^2 + \|\vt\|_2^2 + 2 \langle \vs,\vt\rangle = 2n+ 2 \langle \vs,\vt\rangle.
		\label{eq:trace_redc}
	\end{align*}
	If $\langle \vs,\vt\rangle \geq \sqrt{n}$ then we will have $\tr(\mA) \geq 2(n + \sqrt{n})$ and if $\langle \vs,\vt\rangle \leq -\sqrt{n}$ then we will have $\tr(\mA) \leq 2(n - \sqrt{n})$. So, if Alice and Bob can approximate $\tr(\mA)$ up to relative error $(1\pm 1/\sqrt{n})$, then they can solve \problemref{gap-hamming}. We claim that they can do so with just $O(m\cdot \sqrt{n}(\log n + b))$ bits of communication if there exists an $m$-query adaptive matrix-vector multiplication algorithm for positive semidefinite trace estimation achieving error $(1\pm 1/\sqrt{n})$.

	Specifically, Alice takes charge of running the query algorithm. To compute $\mA\vr$ for a vector $\vr$, Alice and Bob first need to compute $\mZ \vr$. To do so, Alice sends $\vr$ to Bob, which takes $O(\sqrt{n}\cdot b)$ bits since $\vr$ has entries bounded by $2^b$. Bob then computes $\mT \vr$, which has entries bounded by $\sqrt{n}2^b$. He sends the result to Alice, using $O(\sqrt{n}(b+ \log n))$ bits. Upon receiving $\mT \vr$, Alice computes $\mZ \vr = \mS \vr + \mT \vr$. Next, they need to multiply $\mZ\vr$ by $\mZ^T$ to obtain $\mA\vr = \mZ^T\mZ\vr$. To do so, Alice sends $\mZ \vr$ to Bob (again using $O(\sqrt{n}(b+ \log n))$ bits) who computes $\mT^T \mZ\vr$. The entries in this vector are bounded by $2n2^b$, so Bob sends the result back to Alice using $O(\sqrt{n}(b + \log n))$ bits. Finally, Alice computes $\mS^T \mZ\vr$ and adds the result to $\mT^T \mZ\vr$ to obtain $\mZ^T \mZ\vr = \mA \vr$. Given this result, Alice chooses the next query vector according to the algorithm and repeats.
	
	 Overall, running the full matrix-vector query algorithm requires $O(m\cdot \sqrt{n}(\log n + b))$ bits of communication. So, from  \lemmaref{gh_lowerbound} we have that $m = \Omega(\sqrt{n}/(\log n + b))$ queries are needed to approximate the trace to accuracy $1 \pm \epsilon$ for $\epsilon = \nicefrac{1}{\sqrt{n}}$, with probability $> \nicefrac23$.
\end{proof}


\section{Variance Analysis}
\label{sec:variance-analysis}
In this section, we bound the variance of a version of the \hutchpp estimator (\algorithmref{hutchpp-var}), which involves Gaussian random vectors. While the high-probability bounds of \theoremref{intro_thm} and  \theoremref{general-hutchpp-qr} hold for this version of the algorithm, the variance bounds have the advantage of involving (small) explicit constants. They be used to obtain high probability bounds with similarly explicit constants via Chebyshev's inequality, albeit with a worse $\delta$ dependence than \theoremref{intro_thm} and  \theoremref{general-hutchpp-qr}.

\begin{algorithm}[H]
	\caption{\hutchppgauss (Gaussian Variant of \hutchpp)}
	\label{alg:hutchpp-var}
	{\bfseries input}: Matrix-vector multiplication oracle for matrix, \(\mA\in\bbR^{d \times d}\). Number of queries, \(m\). \\
	{\bfseries output}: Approximation to \(\tr(\mA)\).\\
	\vspace{-1em}
	\begin{algorithmic}[1]
		\STATE Sample \(\mS\in\bbR^{d \times \frac{m+2}{4}}\) with \iid \(\cN(0,1)\) entries and \(\mG\in\bbR^{d \times \frac{m-2}{2}}\) with \iid \(\{+1,-1\}\) entries.
		\STATE Compute an orthonormal basis \(\mQ \in \bbR^{d \times \frac{m+2}{4}}\) for the span of \(\mA\mS\) (e.g., via QR decomposition).
		\vspace{-1em}
		\STATE {\bfseries return} \(\hutchppgauss(\mA) = \tr(\mQ^T\mA\mQ) + \frac2{m-2} \tr(\mG^T(\eye-\mQ\mQ^T)\mA(\eye-\mQ\mQ^T)\mG)\).
	\end{algorithmic}
\end{algorithm}
The only difference between \algorithmref{hutchpp-var} and \algorithmref{hutchpp-qr} is that \mS is now Gaussian, and constants are set slightly differently (to minimize variance). Our main result follows:

\begin{theorem}
\label{thm:hutchpp-var}
If \algorithmref{hutchpp-var} is implemented with \(m\) queries, then for PSD \mA,
\[
	\E[\hutchppgauss(\mA)]=\tr(\mA) \hspace{1cm}\text{and}\hspace{1cm} \Var[\hutchppgauss(\mA)] \leq \tsfrac{16}{(m-2)^2}\tr^2(\mA)
\]
\end{theorem}
Before stating the proof, we import three theorems. The first is on the variance of Hutchinson's estimator implemented with Gaussians, and is easy to derive directly:
\begin{importedlemma}[Lemma 1 from \cite{AvronToledo:2011}]
\label{implem:hutch-expect-var}
Hutchinson's estimator implemented with Gaussian random entries has \(\E[H_\ell(\mA)]=\trace(\mA)\) and \(\Var[H_\ell(\mA)] = \frac2\ell \normof{\mA}_F^2\).
\end{importedlemma}
We also require a result on the \textit{expected error} of a randomized low-rank approximation.
In contrast, the proof of \theoremref{intro_thm} uses a \textit{high-probability} result.
Note that the following result is why we use Gaussian random vectors instead of random sign bits.
\begin{importedtheorem}[Theorem 10.5 from \cite{DBLP:journals/siamrev/HalkoMT11}]
\label{impthm:projection-error}
Fix target rank \(k \geq 2\) and oversampling parameter \(p \geq 2\).
Let \(\mS\in\bbR^{d \times (k+p)}\) with \iid \(\cN(0,1)\) entries, and let \(\mQ\in\bbR^{d \times (k+p)}\) be an orthonormal span the columns of \(\mA\mS\).
Then,
\[
	\E\left[\normof{(\eye-\mQ\mQ^T)\mA}^2_F\right]
	\leq
	(1+\tsfrac{k}{p-1})\normof{\mA-\mA_k}_F^2
\]
\end{importedtheorem}
Finally, we state a strengthening of \lemmaref{low-rank-frob-trace}.

\begin{importedlemma}[Lemma 7 from \cite{DBLP:conf/stoc/GilbertSTV07}]
\label{implem:l2-l1-l0-sharp}
Let \(\mA_k\) be the best rank-k approximation to PSD matrix \mA.
Then \(\normof{\mA-\mA_k}_F \leq \frac{1}{2\sqrt k} \trace(\mA)\).
\end{importedlemma}

\begin{proof}[Proof of \theoremref{hutchpp-var}\hspace{-0.25em}]
%
		
Let $q$ be the number of columns in $\mS$ and $\ell$ be the number of columns in $\mG$, constants that will be chosen shortly. Note that \algorithmref{hutchpp-var} uses $m = 2q + \ell$ matrix vector multiplications. 
%
For notational simplicity, let \(\tilde{\tr}(\mA)\defeq\hutchppgauss(\mA)\). We have \(\tilde{\tr}(\mA) = \tr(\mQ^T\mA\mQ) + \hutch_\ell( (\eye-\mQ\mQ^T)\mA(\eye-\mQ\mQ^T))\).

We first prove the unbiased expectation. Note that it suffices to prove that, for any \emph{fixed  $\mQ$}, $\E[\tilde{\trace}(\mA)|\mQ] = \tr(\mA)$. This follows from  cyclic property of trace, the fact that \(\eye-\mQ\mQ^T\) is idempotent, and \importedlemmaref{hutch-expect-var}:
\begin{align*}
	\E[\tilde{\trace}(\mA)|\mQ] & = 	\E[\tr(\mQ^T\mA\mQ) + \hutch_\ell( (\eye-\mQ\mQ^T)\mA(\eye-\mQ\mQ^T)) |\mQ] \\
	&= \trace(\mQ^T\mA\mQ)+ \E[H_\ell((\eye-\mQ\mQ^T)\mA(\eye-\mQ\mQ^T))|\mQ] \\
	&=  \trace(\mQ^T\mA\mQ) + \tr(\eye-\mQ\mQ^T)\mA(\eye-\mQ\mQ^T)) \\
	&= \trace(\mA)
\end{align*}
Then, to bound the variance, we appeal to the Law of Total Variance:
\begin{align}
	\label{eq:total-variance}
	\Var[\tilde{\trace}(\mA)]
	= \E[\Var[\tilde{\trace}(\mA) | \mQ]] + \Var[\E[\tilde{\trace}(\mA)|\mQ]]
\end{align}
The second term is always zero because, as shown above, we always have that $\E[\tilde{\trace}(\mA)|\mQ] = \tr(\mA)$.
So, we only have to bound the first term in \equationref{total-variance}:
\begin{align*}
	\Var[\tilde{\trace}(\mA)]
	= \E[\Var[\tilde{\trace}(\mA) | \mQ]] 
	&= \E\left[\Var[\trace(\mQ^T\mA\mQ) + H_\ell((\eye-\mQ\mQ^T)\mA(\eye-\mQ\mQ^T)) | \mQ]\right] \\
	&= \E\left[\Var[H_\ell((\eye-\mQ\mQ^T)\mA(\eye-\mQ\mQ^T)) | \mQ]\right] \\
	&= \E\left[\tsfrac2\ell \normof{(\eye-\mQ\mQ^T)\mA(\eye-\mQ\mQ^T)}_F^2\right] \tag{\importedlemmaref{hutch-expect-var}}\\
	&\leq \tsfrac2\ell \E\left[\normof{(\eye-\mQ\mQ^T)\mA}_F^2\right] \tag{submultiplicativity} \\
	&\leq \tsfrac2\ell (1+\tsfrac{k}{p-1})\normof{\mA-\mA_k}_F^2 \tag{\importedtheoremref{projection-error}} \\
	&\leq \tsfrac2\ell (1+\tsfrac{k}{p-1})\tsfrac{1}{4k}\trace^2(\mA). \tag{\importedlemmaref{l2-l1-l0-sharp}}
\end{align*}
Following \importedtheoremref{projection-error}, \(k,p \geq 2\) are any values satisfying \(q = k+p\), and the bound is minimized when \(p-1 = k\).
This yields a variance bound of \(\frac{1}{\ell k}\trace^2(\mA)\).
Under the constraint \(m = 2q + \ell\), where \(q = 2k+1\), \(\frac{1}{\ell k}\) is minimized by setting \(k = \frac{m-4}{8}\) and \(\ell = m - \frac{m+4}{2}\), which yields a bound of \(\frac{1}{\ell k} = \frac{16}{(m-2)^2}\).
\end{proof}
Above, only \importedlemmaref{l2-l1-l0-sharp} uses the fact that \mA is PSD.
Furthermore, the proof of \importedlemmaref{l2-l1-l0-sharp} in \cite{DBLP:conf/stoc/GilbertSTV07} actually implies that  \(\normof{\mA-\mA_k}_F \leq \frac{1}{2\sqrt k} \normof{\mA}_*\) for any \mA, where \(\normof{\cdot}_*\) is the nuclear norm.
By repeating the above analysis, we have the following:
\begin{lemma}
\label{lem:non-psd-hutchpp-var}
For any \mA, \algorithmref{hutchpp-var} has
\(
	\E[\hutchppgauss(\mA)] = \trace(\mA)
\)
as well as
\[
	\Var[\hutchpp(\mA)] \leq \tsfrac{8}{m-2}\normof{\mA-\mA_k}_F^2 \leq \tsfrac{16}{(m-2)^2}\normof{\mA}_*^2
\]
where \(k = \frac{m-2}{8}-1\).
\end{lemma}
Like in \theoremref{general-hutchpp-qr}, the first inequality shows how the decay of \mA's eigenvalues impacts the variance of \(\hutchppgauss{}\), while the second inequality is a analog to the variance guarantee in \theoremref{hutchpp-var}.


\section{Experimental Validation}
\label{sec:experiments}
We complement our theory with experiments on synthetic matrices and real-world trace estimation problems. Code for \hutchpp and \hutchppna is available at \url{https://github.com/RaphaelArkadyMeyerNYU/HutchPlusPlus}. We compare these methods to four algorithms, including both our adaptive and non-adaptive methods:
\begin{itemize}
	\item\textbf{Hutchinson's.} The standard estimator run with $\{+1,-1\}$ random vectors.\vspace{-.5em}
	\item\textbf{Subspace Projection.} The method from \cite{SaibabaAlexanderianIpsen:2017}, which computes an orthogonal matrix $\mQ \in \R^{d\times k}$ that approximately spans the top eigenvector subspace of $\mA \in \R^{d\times d}$ and returns $\tr(\mQ^T\mA\mQ)$ as an approximation to $\tr(\mA)$. A similar approach is employed in \cite{Hanyu-Li:2020}. \cite{SaibabaAlexanderianIpsen:2017} computes $\mQ$ using subspace iteration, which requires $k(q+1)$ matrix-vector multiplications when run for $q$ iterations. A larger $q$ results in a more accurate $\mQ$, but requires more multiplications. As in \cite{SaibabaAlexanderianIpsen:2017}, we found that setting $q=1$ gave the best performance, so we did so in our experiments. With $q= 1$, this method is similar to \hutchpp, except that is does not approximate the remainder of the trace outside the top eigenspace.
	\vspace{-.5em}
	\item\textbf{\hutchppbold.} The adaptive method of \algorithmref{hutchpp-qr} with $\{+1,-1\}$ random vectors.\vspace{-.5em}
	\item\textbf{\hutchppnabold.} The non-adaptive method of \algorithmref{hutchpp-streaming} with $c_1=c_3 = \nicefrac14$ and $c_2 = \nicefrac12$ and $\{+1,-1\}$ random vectors.\vspace{-.5em}
\end{itemize}

\subsection{Synthetic Matrices}
We first test the methods above on random matrices with power law spectra. For varying constant $c$, we let \mLambda be diagonal with \(\mLambda_{ii} = i^{-c}\). We generate a random  orthogonal matrix \(\mQ\in\bbR^{5000 \times 5000}\) by orthogonalizing a random Gaussian matrix and set $\mA = \mQ^T \mLambda \mQ$. $\mA$'s eigenvalues are the values in \mLambda. A larger $c$ results in a more quickly decaying spectrum, so we expect Subspace Projection to perform well. A smaller  $c$ results in a slowly decaying spectrum, which will mean that $\|\mA\|_F \ll \tr(\mA)$. In this case, we expect Hutchinson's to outperform its worst case multiplicative error bound: instead of error $\pm \epsilon \tr(\mA)$ after $O(1/\epsilon^2)$ matrix-multiplication queries, \lemmaref{hutch-frob} predicts error on the order of $\pm \epsilon \|\mA\|_F$. Concretely, for dimension $d = 5000$ and $c = 2$, we have $\|\mA\|_F = .63\cdot \tr(\mA)$ and for $c = .5$ we have $\|\mA\|_F = .02\cdot \tr(\mA)$. In general, unlike the Subspace Projection method and Hutchinson's estimator, we expect \hutchpp and \hutchppna to be less sensitive to $\mA$'s spectrum. 

In \figureref{synth_fig} we plot results for various $c$. Relative error should scale roughly as $\epsilon = O(m^{-\gamma})$, where $\gamma = \nicefrac12$ for Hutchinson's and $\gamma = 1$ for \hutchpp and \hutchppna. We thus use log-log plots, where we expect a linear relationship between the error $\epsilon$ and number of iterations $m$.

The superior performance of \hutchpp and \hutchppna shown in \figureref{synth_fig} is not surprising. These methods are designed to achieve the ``best of both worlds'': when $\mA$'s spectrum decays quickly, our methods approximate $\tr(\mA)$ well by projecting off the top eigenvalues. When it decays slowly, they perform essentially no worse than Hutchinson's. 
We note that the adaptivity of \hutchpp leads to consistently better performance over \hutchppna, and the method is simpler to implement as we do not need to set the constants $c_1, c_2, c_3$. Accordingly, this is the method we move forward with in our real data experiments.

\begin{figure}[t]
	\centering
	\begin{subfigure}{0.43\columnwidth}
		\includegraphics[width=\columnwidth]{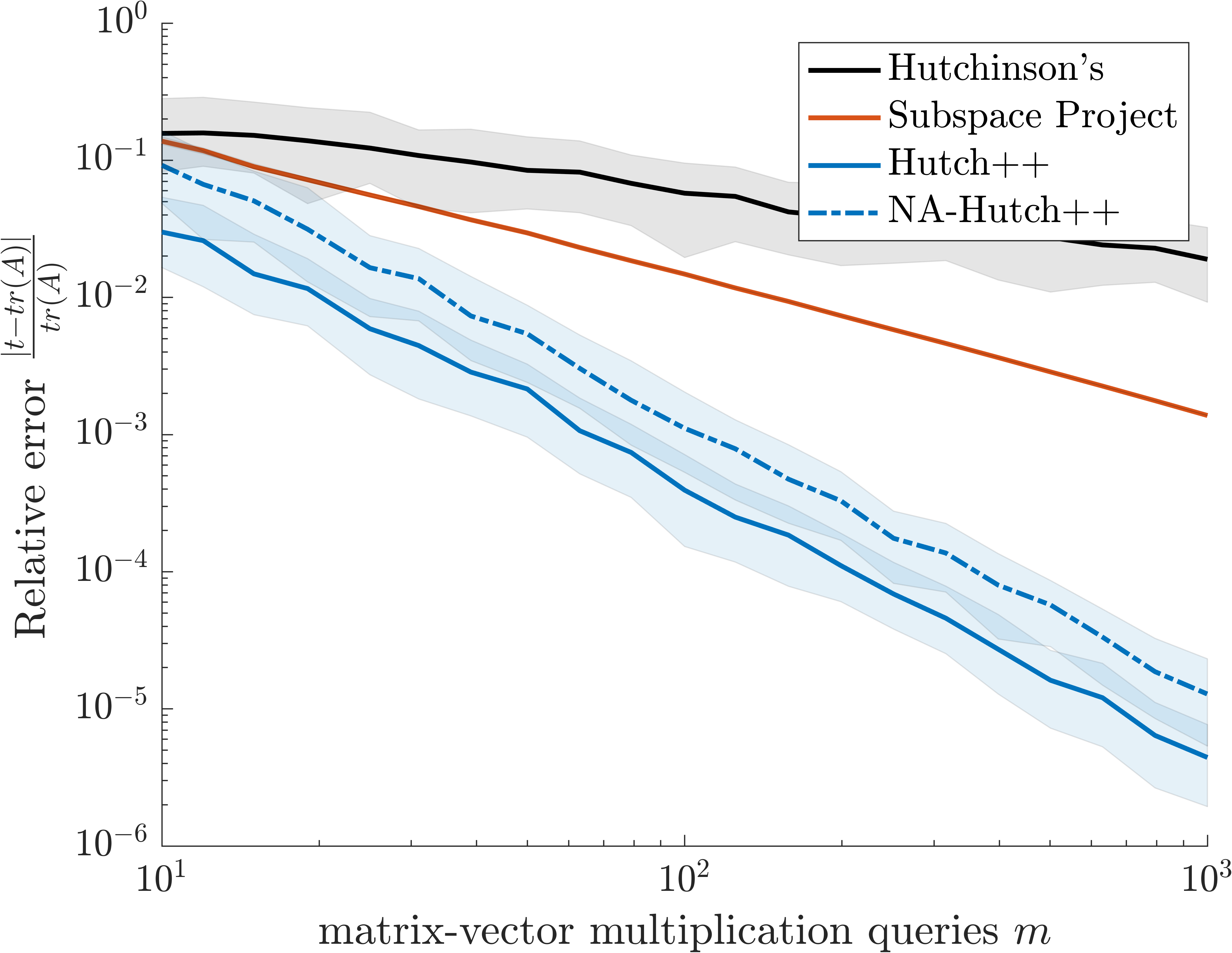}
		\caption{Fast Eigenvalue Decay ($c = 2$)}
	\end{subfigure}\hfill
	\begin{subfigure}{0.43\columnwidth}
		\includegraphics[width=\columnwidth]{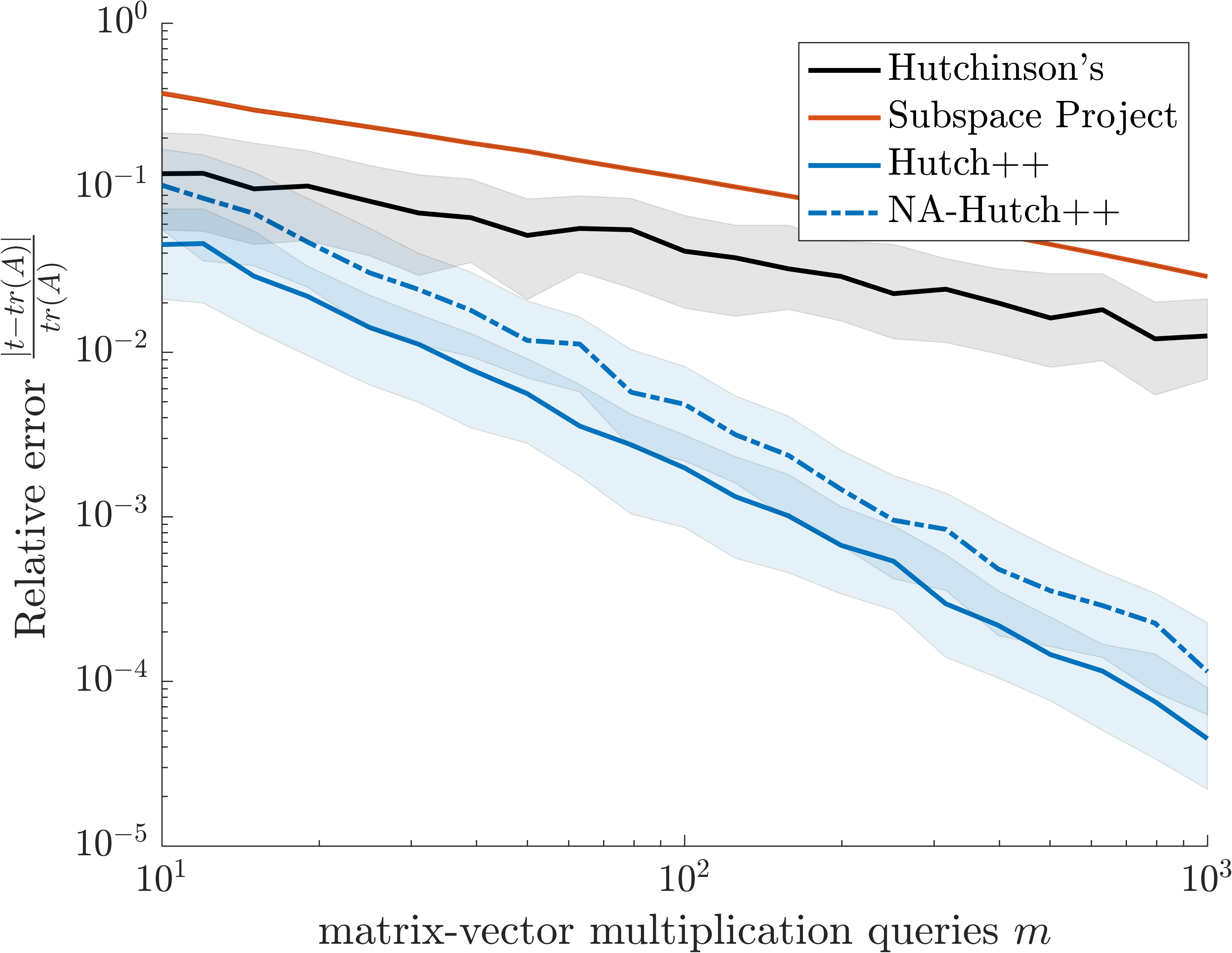}
		\caption{Medium Eigenvalue Decay ($c = 1.5$)}
	\end{subfigure}
	
	\begin{subfigure}{0.43\columnwidth}
		\includegraphics[width=\columnwidth]{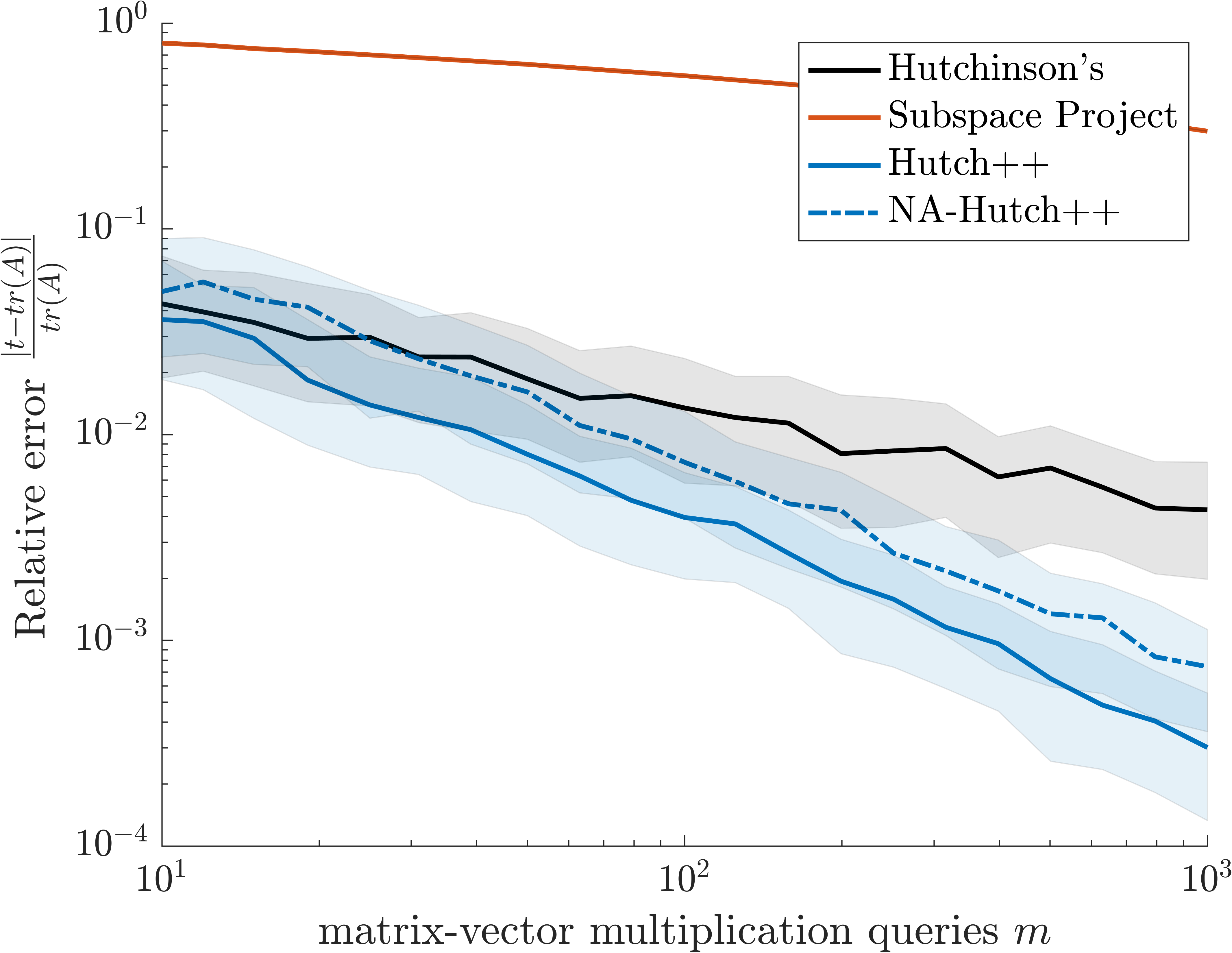}
		\caption{Slow Eigenvalue Decay ($c = 1$)}
	\end{subfigure}\hfill
	\begin{subfigure}{0.43\columnwidth}
		\includegraphics[width=\columnwidth]{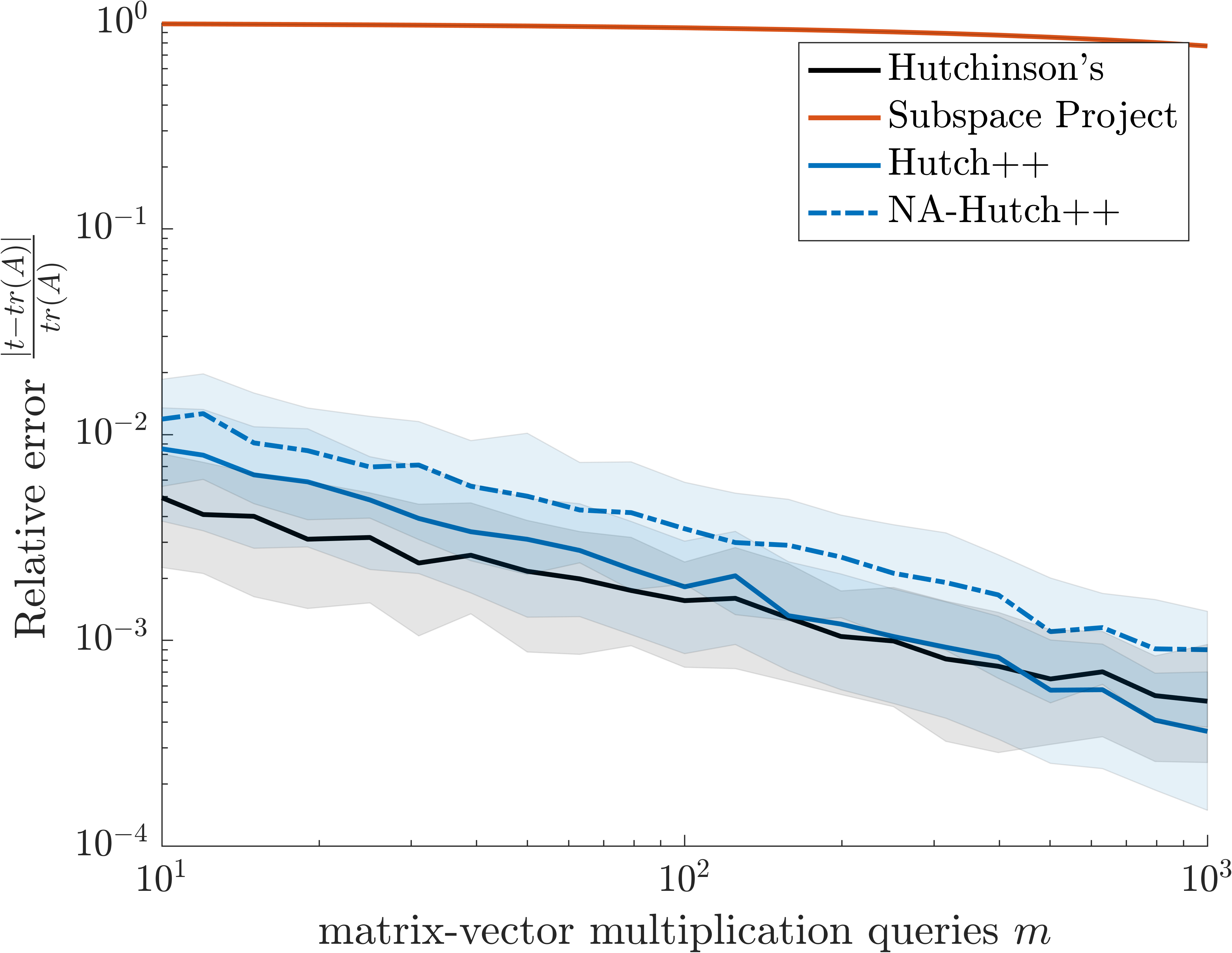}
		\caption{Very Slow Eigenvalue Decay ($c = .5)$}
	\end{subfigure}
	\caption{
		Relative error versus number of matrix-vector multiplication queries for trace approximation algorithms run on random matrices with power law spectra. We report the median relative error of the approximation $t$ after $200$ trials. The upper and lower bounds of the shaded region around each curve are the $25^\text{th}$ and $75^\text{th}$ percentile errors. Subspace Projection has consistently low variance, but as expected, only performs better than Hutchinson's when $c=2$ and there is very fast eigenvalue decay. \hutchpp and \hutchppna typically outperform both methods.
	}
	\label{fig:synth_fig}
\end{figure}

\subsection{Real Matrices}
To evaluate the real-world performance of \hutchpp we test it in the common setting where $\mA = f(\mB)$. In most applications, $\mB$ is symmetric with eigendecomposition $\mV^T\mLambda \mV$, and $f: \R\rightarrow \R$ is a function on real valued inputs. Then we have 
$f(\mB) = \mV^Tf(\mLambda) \mV$ where $f(\mLambda)$ is simply $f$ applied to the real-valued eigenvalues on the diagonal of $\mLambda$. When $f$ returns negative values, $\mA$ may not be postive semidefinite. Generally, computing $f(\mB)$ explicitly requires a full eigendecomposition and thus $\Omega(n^3)$ time. However, many iterative methods can more quickly approximate matrix-vector queries of the form $\mA \vr = f(\mB) \vr$. The most popular and general is the Lanczos method, which we employ in our experiments \cite{UbaruChenSaad:2017,MuscoMuscoSidford:2018}.\footnote{We use the implementation of Lanczos available at \href{https://github.com/cpmusco/fast-pcr}{https://github.com/cpmusco/fast-pcr}, but modified to block matrix-vector multiplies when run on multiple query vectors.}

We consider trace estimation in three example applications, involving both PSD and non-PSD matrices. We test on relatively small inputs, for which we can explicitly compute $\tr(f(\mB))$ to use as a baseline for the approximation error. However, our methods can scale to much larger matrices.

\begin{figure}[h]
	\centering
	\begin{subfigure}{0.45\columnwidth}
		\includegraphics[width=\columnwidth]{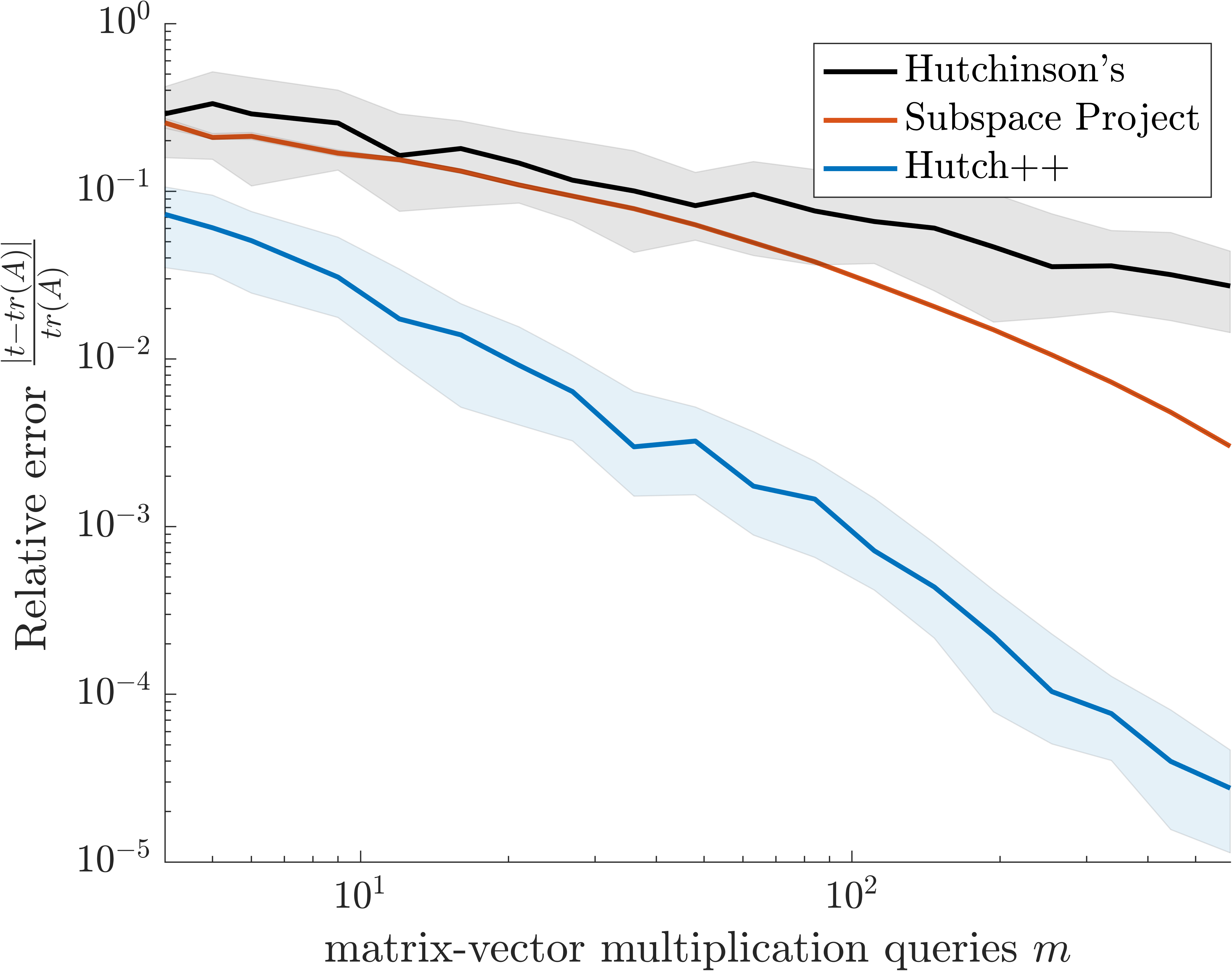}
		\caption{$\mA = \exp(\mB)$, where $\mB$ is the Roget's Thesaurus semantic graph adjacency matrix. For use in Estrada index computation. $\mA$ is PSD.}
	\end{subfigure}\hfill
	\begin{subfigure}{0.45\columnwidth}
		\includegraphics[width=\columnwidth]{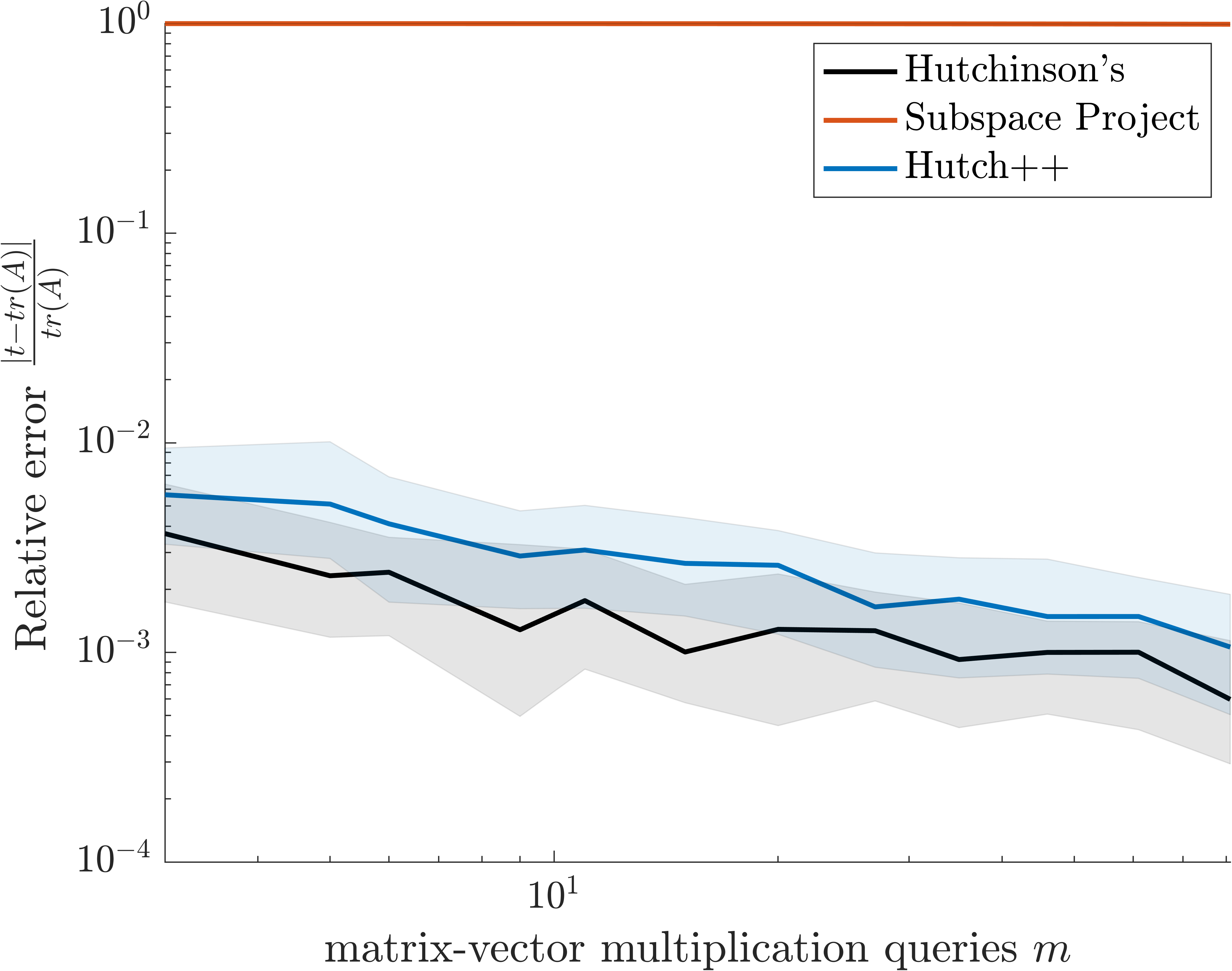}
		\caption{$\mA = \log(\mB + \lambda \eye)$, where $\mB$ is a 2D Gaussian process kernel covariance matrix. For use in log-likelihood computation.  $\mA$ is not PSD.}
	\end{subfigure}
	
	\caption{
		Relative error versus number of matrix-vector multiplication queries for trace approximations of transformed matrices, which were multiplied by vectors using the Lanczos method. We report median relative error of the approximation $t$ after $100$ trials. The upper and lower bounds of the shaded region around each curve are the $25^\text{th}$ and $75^\text{th}$ percentile errors. As expected, Subspace Project and \hutchpp outperform Hutchinson's when $\mA = \exp(\mB)$, as exponentiating leads to a quickly decaying spectrum. On the other hand, Hutchinson's performs well for $\mA = \log(\mB + \lambda \eye)$, which has a very flat spectrum. \hutchpp is still essentially as fast. Subspace Project fails in this case because the top eigenvalues of $\mA$ do not dominate its trace.
	}
	\label{fig:exp_fig}
\end{figure}
\medskip\noindent\textbf{Graph Estrada Index.} Given the binary adjacency matrix $\mB \in \{0,1\}^{d \times d}$ of a graph $G$, the Estrada index is defined as $\tr(\exp(\mB))$ \cite{estrada2000characterization,de2007estimating}, where $\exp(x) = e^x$. This index measures the strength of connectivity within $G$. A simple transformation of the Estrada index yields the \emph{natural connectivity} metric,  defined as $\log \left (\frac{1}{d} \tr(\exp(\mB)) \right )$ \cite{jun2010natural,EstradaHatanoBenzi:2012}. 

In our experiments, we approximated the Estrada index of the Roget's Thesaurus semantic graph, available from  \cite{BatageljMrvar:2006}. The Estrada index of this $1022$ node graph was originally studied in \cite{EstradaHatano:2008}. We use the Lanczos method to approximate matrix multiplication with $\exp(\mB)$, running it for $40$ iterations, after which the error of application was negligible compared to the approximation error of trace estimation. Results are shown in \figureref{exp_fig}.

\begin{figure}[h]
	\centering
	\begin{subfigure}{0.45\columnwidth}
		\includegraphics[width=\columnwidth]{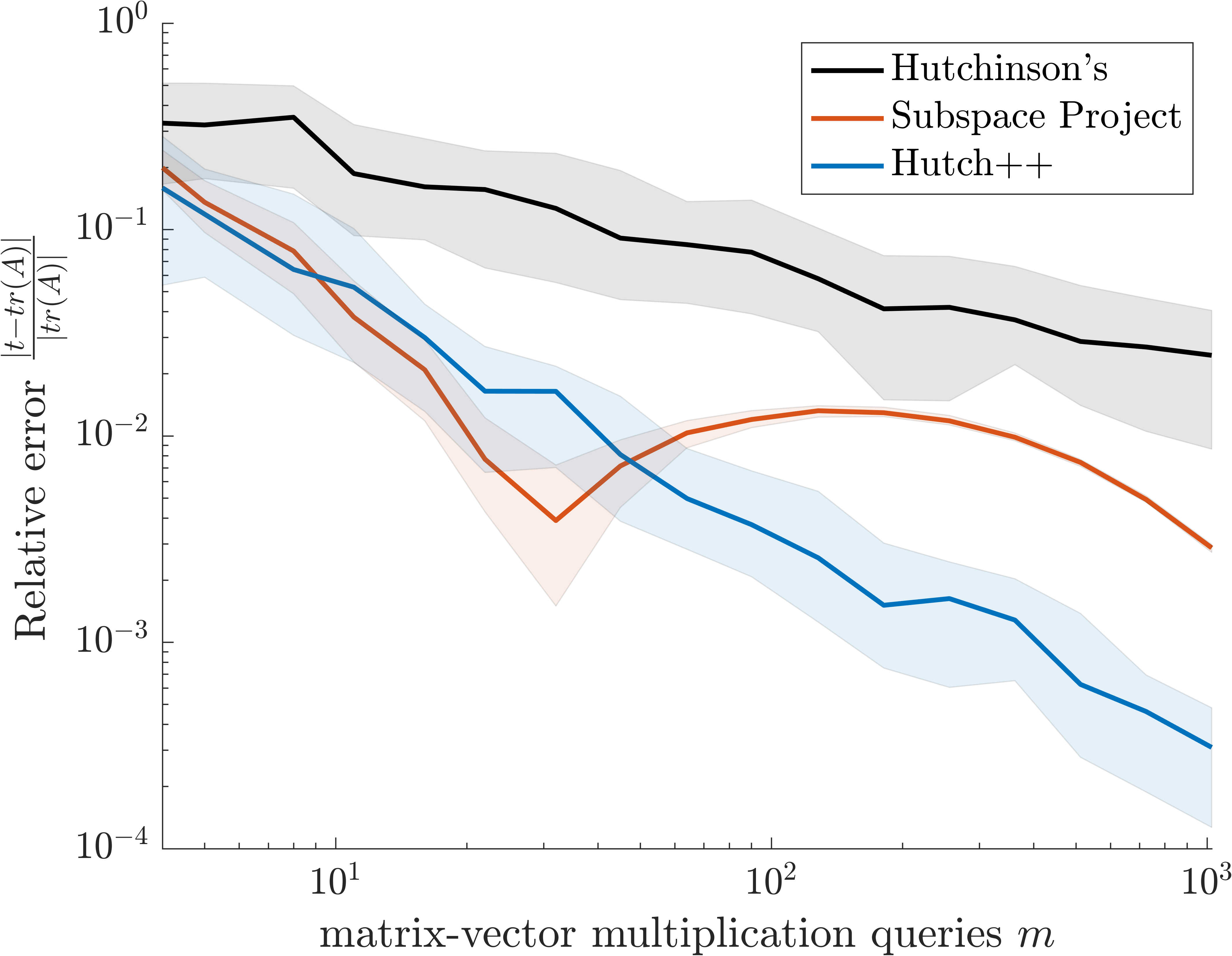}
		\caption{$\mA = \mB^3$ where $\mB$ is a Wikipedia voting network adjacency matrix. For use in triangle counting. $\mA$ is not PSD.}
	\end{subfigure}\hfill
	\begin{subfigure}{0.45\columnwidth}
		\includegraphics[width=\columnwidth]{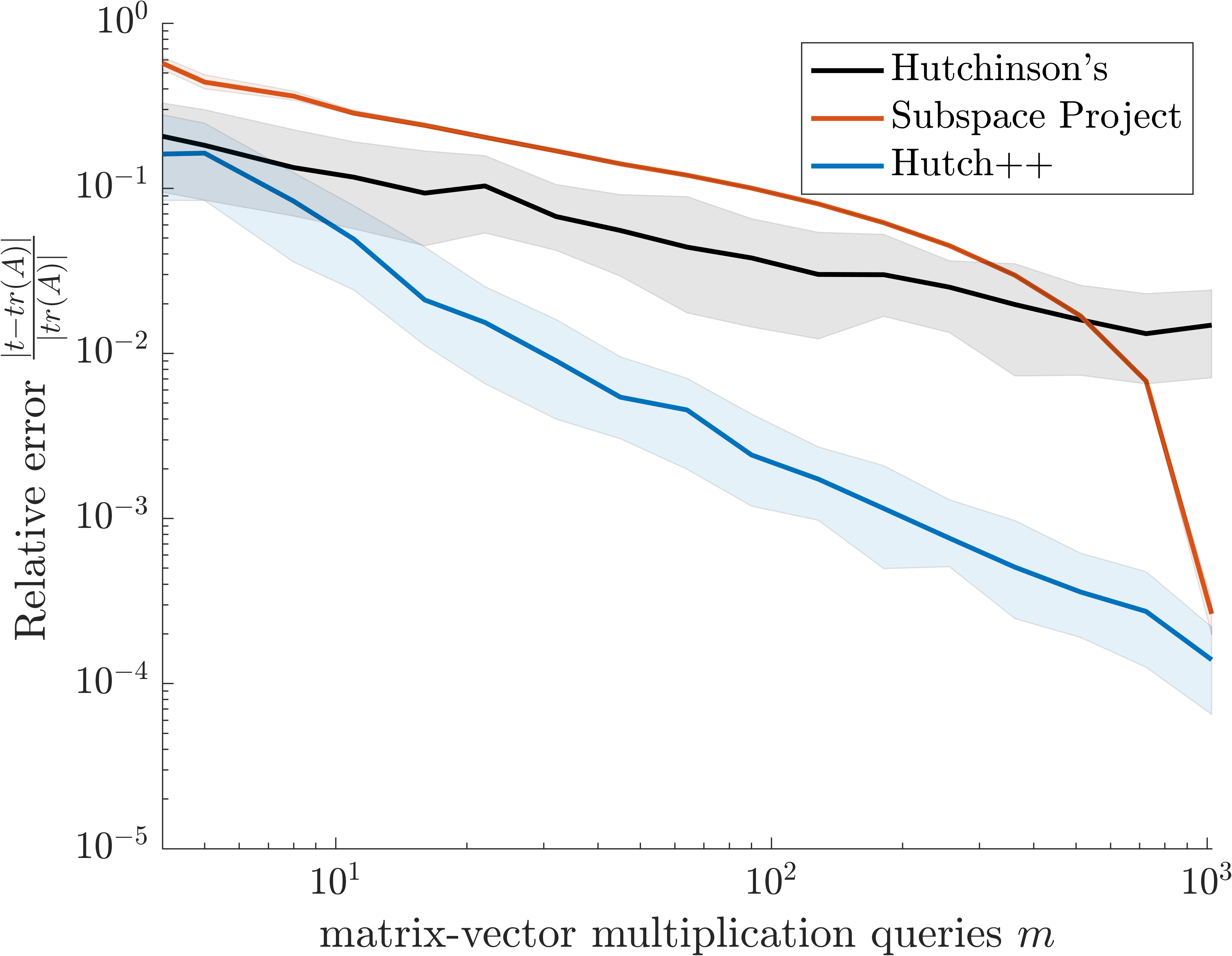}
		\caption{$\mA = \mB^3$ where $\mB$ is an arXiv.org citation network adjacency matrix. For use in triangle counting. $\mA$ is not PSD.}
	\end{subfigure}
	
	\caption{
		Relative error versus number of matrix-vector multiplication queries for trace approximations of transformed matrices. We report the median relative error of the approximation $t$ after $100$ trials. The upper and lower bounds of the shaded region around each curve are the $25^\text{th}$ and $75^\text{th}$ percentile errors.
		\hutchpp still outperforms the baseline methods even though $\mA$ is not PSD. We note that Subspace Project has somewhat uneven performance: increasing $m$ will take into account a larger number of top eigenvalues when approximating the trace. However, since these may be positive or negative, approximation error does not monotonically decrease. \hutchpp is not sensitive to this issue since it does not use \emph{just} the top eigenvalues: see \figureref{triangle__spect_fig} for more discussion.
	}
	\label{fig:triangle_fig}
\end{figure}

\medskip\noindent\textbf{Gaussian Process Log Likelihood.} Let $\mB\in\R^{d\times d}$ be a PSD kernel covariance matrix and let $\lambda \ge 0$ be a regularization parameter. In Gaussian process regression, the model log likelihood computation requires computing $\log\det(\mB + \lambda \eye) = \tr(f(\mB))$ where $f(x) = \log(x + \lambda)$ \cite{williams1996gaussian,rasmussen2004gaussian}. This quantity must be computed repeatedly for different choices of $\mB$ and $\lambda$ during hyper-parameter optimization, and it is often approximated using Hutchinson's method  \cite{BoutsidisDrineasKambadur:2015,UbaruChenSaad:2017,HanMalioutovAvron:2017,ScalableLD}. We note that, while $\mB$ is positive semidefinite, $\log(\mB + \lambda \eye)$ typically will not be. So our strongest theoretical bounds do not apply in this case, but \hutchpp can be applied unmodified, and as we see in \figureref{exp_fig}, still gives good performance.

In our experiments we consider a benchmark 2D Gaussian process regression problem from the GIS literature, involving precipitation data from Slovakia \cite{NetelerMitasova:2013}. $\mB$ is the kernel covariance matrix on $6400$ randomly selected training points out of 196,104 total points. Following the setup of \cite{erdelyi2020fourier}, we let $\mB$ be a Gaussian kernel matrix with width parameter $\gamma = 64$ and regularization parameter $\lambda = .008$, both determined via cross-validation on $\ell_2$ regression loss.

\begin{figure}[H]
	\centering
	\begin{subfigure}{0.45\columnwidth}
		\includegraphics[width=\columnwidth]{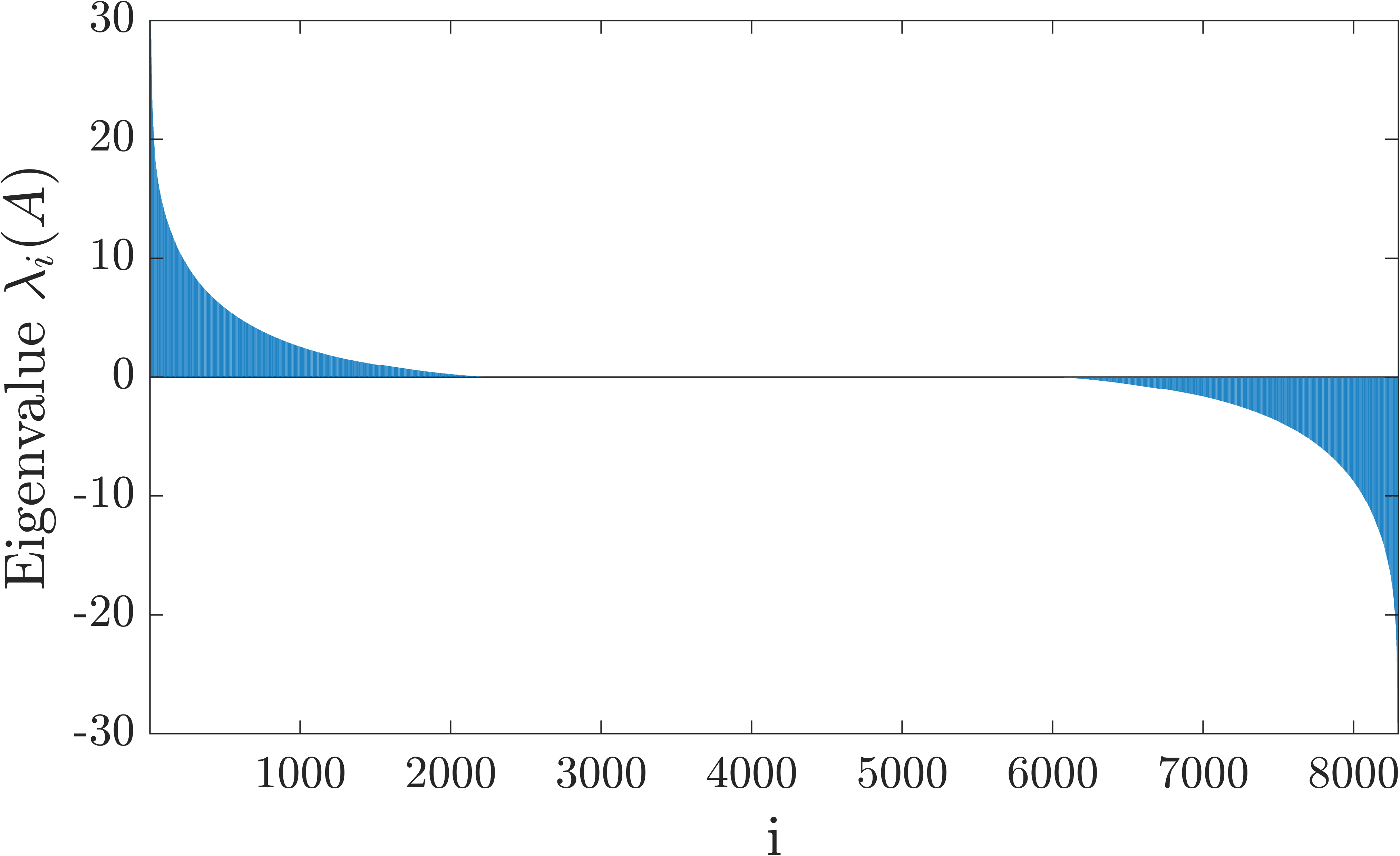}
		\caption{Eigenvalues for Wikipedia voting network adjacency matrix $\mB$.}
	\end{subfigure}\hfill
	\begin{subfigure}{0.45\columnwidth}
		\includegraphics[width=\columnwidth]{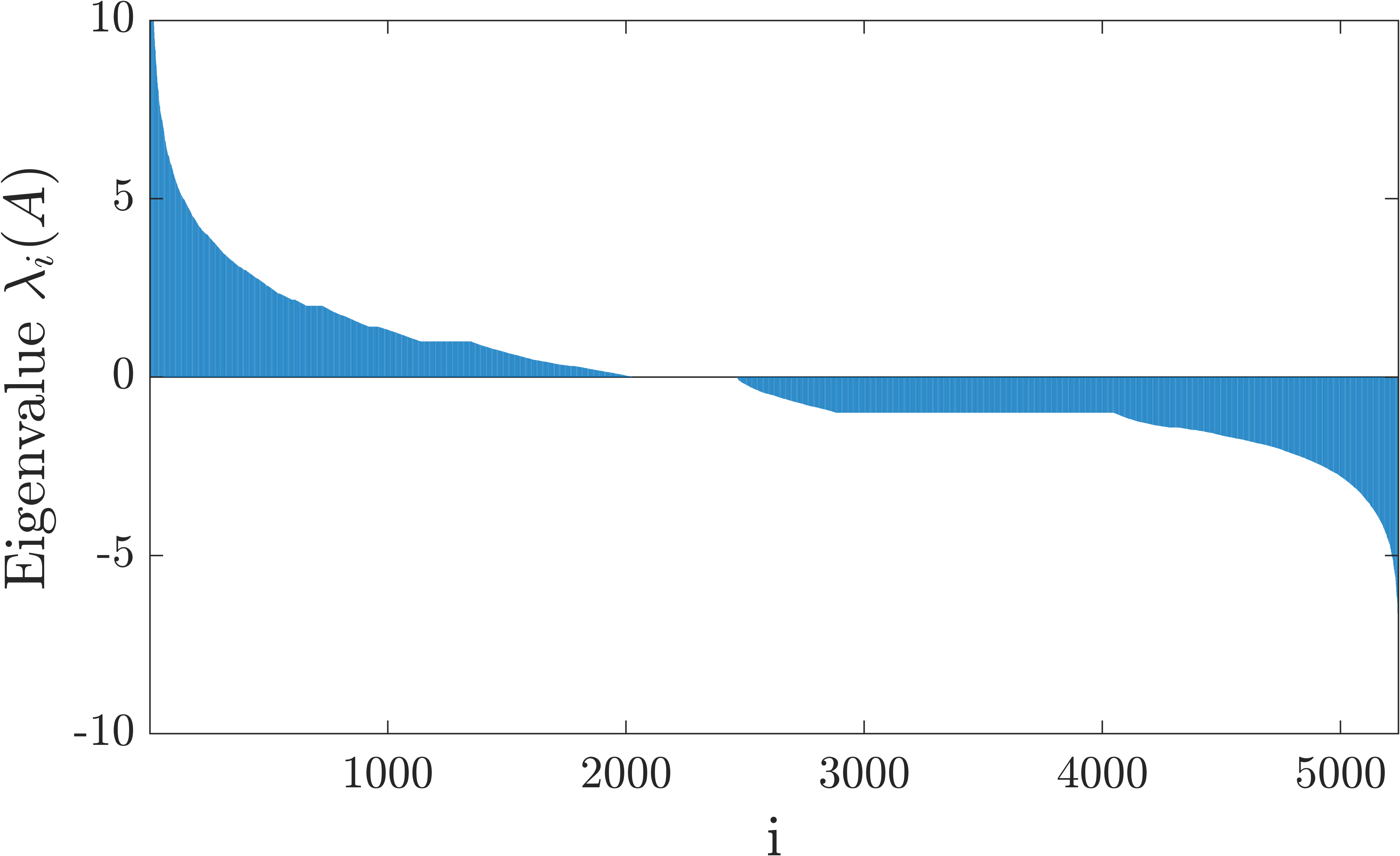}
		\caption{Eigenvalues for arXiv.org citation network adjacency matrix $\mB$.}
	\end{subfigure}
	
	\caption{
		Even when estimating the trace of a non-PSD matrix like $\mA = \mB^3$, which for the triangle counting examples above will have both positive and negative eigenvalues, \hutchpp can far outperform Hutchinson's method.
		As \theoremref{general-hutchpp-qr} and \lemmaref{non-psd-hutchpp-var} both suggest, it will approximately project off the largest \emph{magnitude} eigenvalues from $\mA$ (whether postive or negative), which then reduces the variance of estimating the trace $\tr(\mA) = \sum_{i=1}^d \lambda_i(\mB)^3$.   
	}
	\label{fig:triangle__spect_fig}
\end{figure}

\medskip\noindent\textbf{Graph Triangle Counting.} Given the binary adjacency matrix $\mB \in \{0,1\}^{d \times d}$ of an undirected graph $G$, the number of triangles in $G$ is equal to $\frac{1}{6}\tr(\mB^3)$. The triangle count is an important measure of local connectivity and extensive research studies its efficient approximation \cite{schank2005finding,becchetti2008efficient,pagh2012colorful}. Popular approaches include applying Hutchinson's method to $\mA = \mB^3$ \cite{Avron:2010}, or using the {EigenTriangle} estimator, which is similar to the Subspace Projection method \cite{tsourakakis2008fast}.

In our experiments, we study approximate triangle counting on two common benchmark graphs: an arXiv.org collaboration network\footnote{Link: \url{https://snap.stanford.edu/data/ca-GrQc.html}.} with 5,243 nodes and 48,260 triangles, and a Wikipedia administrator voting network\footnote{Link: \url{https://snap.stanford.edu/data/wiki-Vote.html}.} with 7,115 nodes and 608,389 triangles.
We again note that the adjacency matrix $\mB$ is not positive semidefinite, and neither is $\mA = \mB^3$. Nevertheless, we can apply \hutchpp and see very strong performance. In this setting we do not need to apply Lanczos for matrix-vector query computation: $\mA \vr$ can be computed exactly  using three matrix-vector multiplications with $\mB$. Results are shown \figureref{triangle_fig} with graph spectral visualized in \figureref{triangle__spect_fig}

\section*{Acknowledgments}
The authors would like to thank Joel A. Tropp for suggesting the variance analysis in \sectionref{variance-analysis}, as well as other valuable comments on the paper.
D. Woodruff would like to thank support from the National Institute of Health (NIH) grant 5R01 HG 10798-2 and a Simons Investigator Award.

\bibliographystyle{alpha}
\bibliography{hutchplusplus}

\appendix

\section{Proof of \lemmaref{hutch-frob}}
\label{app:hutch-frob}

We start by stating the Hanson-Wright inequality for i.i.d sub-Gaussian random variables:
\begin{importedtheorem}[\cite{rudelson2013hanson}]
\label{impthm:hanson-wright}
Let \(\vx\in\R^n\) be a vector of mean 0, i.i.d. sub-Gaussian random variables with constant sub-Gaussian parameter $C$.
Let \(\mA\in\bbR^{n \times n}\) be a matrix.
Then, there exists a constant \(c\) only depending on \(C\)  such that for every \(t \geq 0\),
\begin{align*}
	\Pr\left\{\abs{\vx^T\mA\vx - \E[\vx^T\mA\vx]} > t\right\}
	\leq 2 \exp\left(-c\cdot \min\left\{\frac{t^2}{\normof{\mA}_F^2} , \frac{t}{\normof{\mA}_2}\right\}\right).
\end{align*}
\end{importedtheorem}
Above, $\|\mA\|_2 = \max_x \|\mA\vx\|_2/\|\vx\|_2$ denotes the spectral norm.
We refer the reader to \cite{rudelson2013hanson} for a formal definition of sub-Gaussian random variables: both normal $\cN(0,1)$ random variables and $\pm 1$ random variables are sub-Gaussian with constant $C$.

\begin{replemma}{hutch-frob}
	Let \(\mA\in\bbR^{d \times d}\), \(\delta\in(0, \nicefrac12]\), $\ell\in\bbN$. Let  $\hutch_\ell(\mA)$ be the  \(\ell\)-query Hutchinson estimator defined in \eqref{eq:hutch}, implemented with mean 0, i.i.d. sub-Gaussian random variables with constant sub-Gaussian parameter. For fixed constants $c,C$, if $\ell > c\log(\nicefrac1\delta)$, then with prob. \(1-\delta\), 
\[
\abslr{\hutch_\ell(\mA) - \trace(\mA)} \leq C\sqrt{\frac{\log(\nicefrac1\delta)}{\ell}} \normof\mA_F.
\]
\end{replemma}

\begin{proof}
Let \(\bar\mA\in\bbR^{\ell d \times \ell d}\) be a block-diagonal matrix formed from $\ell$ repetitions of \(\mA\):
\[
	\bar\mA \defeq \bmat{
		\mA   & \mat0 & \ldots &\mat0  \\
		\mat0 & \mA   & \ldots &\mat0  \\
		\mat0 & \mat0 & \ddots &\vdots \\
		\mat0 & \mat0 & \ldots &\mA    \\
	}.
\]
Let $\mG \in \R^{d\times \ell}$ be as in \eqref{eq:hutch}. Let \(\vg_i\) be $\mG$'s \(i^\text{th}\) column and let \(\vg = [\vg_1, \ldots,\vg_\ell] \in\bbR^{d\ell}\) be a vectorization of \mG. We have that \(\ell\cdot\hutch_\ell(\mA) = \trace(\mG^T\mA\mG) = \vg^T\bar\mA\vg\).
So, by \importedtheoremref{hanson-wright},
\begin{align}
\label{eq:hanson_wright2}
	\Pr\left\{\abs{\vg^T\bar\mA\vg - \E[\vg^T\bar\mA\vg]} > t\right\}
	\leq 2 \exp\left(-c \cdot \min\left\{\frac{t^2}{\normof{\bar\mA}_F^2} , \frac{t}{\normof{\bar\mA}_2}\right\}\right).
\end{align}
We let $t' = t/\ell$, and substitute \(\E[\vg^T\bar\mA\vg]=\trace(\bar\mA)=\ell\trace(\mA)\), \(\normof{\bar\mA}_F^2 = \ell\normof{\mA}_F^2\), and \(\normof{\bar\mA}_2 = \normof{\mA}_2\) into \eqref{eq:hanson_wright2} to get:
\begin{align*}
	\Pr\left\{\abslr{\hutch_\ell(\mA) - \trace(\mA)} > t'\right\}
	\leq 2 \exp\left(-c \min\left\{\frac{\ell t'^2}{\normof{\mA}_F^2} , \frac{\ell t'}{\normof{\mA}_2}\right\}\right).
\end{align*}
Now, taking \(t' = \sqrt{\frac{\ln(\nicefrac2\delta)}{c\ell}} \normof\mA_F\), we have:
\begin{align*}
	\Pr\left\{\abslr{\hutch_\ell(\mA) - \trace(\mA)} > \sqrt{\frac{\ln(\nicefrac2\delta)}{c\ell}} \normof\mA_F \right\}
	\leq 2 \exp\left(- \min\left\{\log(\nicefrac2\delta), \sqrt{c\ell\log(\nicefrac2\delta)} \frac{\normof{\mA}_F}{\normof\mA_2}\right\}\right)	
\end{align*}
Since \(\frac{\normof\mA_F}{\normof\mA_2} \geq 1\),  if we take $\ell \geq  \ln(\nicefrac2\delta)/c $, we have that the minimum takes value \(\log(\nicefrac2\delta)\), so 
\[
	\Pr\left\{\abslr{\hutch_\ell(\mA) - \trace(\mA)} > \sqrt{\frac{\ln(\nicefrac2\delta)}{c\ell}} \normof\mA_F\right\}
	\leq \delta.
\]
The final result follows from noting that $\ln(\nicefrac2\delta) \leq 2\ln(\nicefrac1\delta)$ for $\delta \leq 1/2$.
\end{proof}

\section{Proof of \theoremref{nonadaptive_lower_bound}}
\label{app:lower_bound}
To prove our non-adaptive lower bound for the real RAM moodel we
first introduce a simple testing problem which we reduce to estimating the trace of a PSD matrix \mA to \((1\pm\epsilon)\) relative error:

\begin{problem}
	\label{prob:trace-style}
	Fix \(d,n\in\bbN\) such that $d \geq n$ and \(n\defeq\frac1\eps\) for $\epsilon \in (0,1]$.
	Let $\mD_1 = \mI_n$ and $\mD_2 = \spmat{\mI_{n-1} & \\ & 0}$.\footnote{Here $\mI_r$ denotes an $r \times r$ identity matrix.} Consider \(\mA = \mG^T\mD\mG\) generated by selecting \(\mG\in\bbR^{n \times d}\) with \iid random Guassian \(\cN(0,1)\) entries and \(\mD=\mD_1\) or \(\mD=\mD_2\) with equal probability.
	Then consider any algorithm which fixes a query matrix \(\mU\in\bbR^{d \times m}\), observes \(\mA\mU\in\bbR^{d \times m}\), and guesses if \(\mD=\mD_1\) or \(\mD=\mD_2\).
\end{problem}
The reduction from \problemref{trace-style} to relative error trace estimation is as follows: 

\begin{lemma}
	\label{lem:lower-bound-trace-style}
	For any $\epsilon \in (0,1]$ and sufficient large $d$, if a randomized algorithm \cA can estimate the trace of any $d \times d$ PSD matrix to relative error  \(1 \pm \frac\eps4\) with success probability \(\geq \frac34\) using $m$ queries, then \cA can be used to solve \problemref{trace-style} with success  probability \(\geq\frac23\) using $m$ queries.
\end{lemma}
\begin{proof}
	To solve \problemref{trace-style} we simply apply \cA to the matrix $\mA = \mG^T\mD \mG$ and guess $\mD_1$ if the trace is closer to $\frac{d}{\epsilon}$ and $\mD_2$ if it's closer to $\frac{d}{\epsilon}-d$. To see that this succeeds with probability $2/3$, we first need to understand the trace of $\mA$. To do so, note that $\tr(\mA) = \tr(\mG^T\mD \mG)$ is simply a scaled Hutchinson estimate for $\tr(\mD)$, i.e.
	$\trace(\mG^T\mD\mG) = d \cdot \hutch_d(\mD)$. So, via \lemmaref{hutch-frob}, for large enough $d$ we have that with probability $\ge \frac{11}{12}$ both of the following hold:
	\begin{align*}
	\frac1d\trace(\mG^T\mD_1\mG) \geq \left (1-\frac{\epsilon}{4}\right )\trace(\mD_1)
	\hspace{0.5cm}\text{ and}\hspace{0.5cm}
	\frac1d\trace(\mG^T\mD_2\mG) \leq \left(1+\frac{\epsilon}{4}\right )\trace(\mD_2).
	\end{align*}
	Additionally, with probability \(\frac34\), \cA computes an approximation $Z$ with $(1-\frac\eps4) \tr(\mA) \le Z \le (1+\frac\eps4) \tr(\mA)$.
	By a union bound, all of the above events happen with probability $\ge \frac23$.
	If \(\mD=\mD_1\):
	\[
	Z \geq (1-\tsfrac\eps4) \tr(\mA) \geq (1-\tsfrac\eps4)^2 \cdot d \cdot \trace(\mD_1) > (1-\tsfrac{\epsilon}{2})\cdot \frac{d}{\epsilon}.
	\]
	On the other hand, if \(\mD=\mD_2\),
	\begin{align*}
	Z
	\leq (1+\tsfrac\eps4) \tr(\mA)
	&\leq (1+\tsfrac\eps4)^2 \cdot d \cdot \trace(\mD_2)\\
	&= (1+\tsfrac\eps4)^2 \cdot (1-\eps)\cdot \frac{d}{\epsilon}
	< (1-\tsfrac{\epsilon}{2})\cdot \frac{d}{\epsilon}.
	\end{align*}
	Thus, with probability $2/3$, $Z$ is closer to $\frac{d}{\epsilon}$ when $\mD = \mD_1$ and closer to $\frac{d}{\epsilon}-d$ when $\mD = \mD_2$, so the proposed scheme guesses correctly.
\end{proof}

In the remainder of the section we show that \problemref{trace-style} requires $\Omega(1/\epsilon)$ queries, which combined with \lemmaref{lower-bound-trace-style} proves our main lower bound, \theoremref{nonadaptive_lower_bound}.
Throughout,  we let $\mX \eqdist \mY$ denote  that $\mX$ and $\mY$ are identically distributed.
We first argue that for \problemref{trace-style}, the non-adaptive query matrix $\mU$ might as well be chosen to be the first $m$ standard basis vectors.

\begin{lemma}
	\label{lem:remove-query}
	For \problemref{trace-style}, without loss of generality, we may assume that the query matrix \mU equals \(\mU = \mE_m=\bmat{\ve_1 & \ldots & \ve_m}\), the first \(m\) standard basis vectors.
\end{lemma}
\begin{proof}
	First, we may assume without loss of generality that $\mU$ is orthonormal, since if it were not, we could simply reconstruct the queries $\mA \mU$ by querying $\mA$ with an orthonormal basis for the columns of $\mU$. 
	Next, by rotational invariance of the Gaussian distribution, if \(\mG \in \R^{n \times d}\) is an \iid \(\cN(0,1)\) matrix, and \(\mQ \in \R^{d \times d}\) is any orthogonal matrix, then \(\mG\mQ\) is  distributed identically to $\mG$.
	Let  \(\bar\mU \in \R^{d \times d - m}\) be any orthonormal span for the nullspace of \mU, so that \(\mQ\defeq\bmat{\mU & \bar\mU}\) is orthogonal.
	We have that  \(\mQ \mG^T\mD\mG \mE_m \eqdist \mQ  \mQ^T \mG^T\mD\mG \mQ \mE_m = \mG^T\mD\mG \mU\).
	So, using the result $\mG^T\mD\mG \mE_m$ of querying with matrix $\mE_m$, we can just multiply by $\mQ$ on the left to obtain a set of vectors that has the same distribution as if $\mU$ had been used as a query matrix.
\end{proof}

With \lemmaref{remove-query} in place, we are able to reduce \problemref{trace-style} to a simpler testing problem on distinguishing $m$ random vectors drawn from normal distributions with different covariance matrices:

\begin{problem}
	\label{prob:statistical-estimation}
	Let \(n=\frac1\eps\) and let $\vz\in\R^n$ be a uniformly random unit vector. Let $\mN \in \R^{n\times m}$ contain $m$ i.i.d. random Gaussian vectors drawn from an $n$-dimensional Gaussian distribution, $\cN(\mat{0},\mC)$, where the covariance matrix $\mC$ either equals $\mI$ or $\mI - \vz\vz^T$, with equality probability. The goal is to use $\mN$ to distinguish, with probability $> \tfrac{1}{2}$ what the true identity of $\mC$ is. 
\end{problem}

\begin{lemma}
	\label{lem:reduce-to-statistical-estimation}
	Let \cA be an algorithm that solves \problemref{trace-style} with $m$ queries and success probability $p$.
	Then \cA can be used to solve \problemref{statistical-estimation} with $m$ Gaussian samples and the same success probability.
\end{lemma}
\begin{proof}
	By \lemmaref{remove-query}, it suffices to show how to use the observed matrix \(\mN\) in \problemref{statistical-estimation} to create a sample from the distribution \(\mG^T\mD\mG \mE_m\) where $\mG \in \R^{n \times d}$ has \iid \(\cN(0,1)\) entries. Specifically, we claim that, if we sample \(\mL\in\bbR^{n \times (d-m)}\) with \iid \(\cN(0,1)\) entries, and compute
	\[
	\mM = \bmat{\mN^T\mN \\ \mL^T\mN},
	\]
	then $\mM$ is identically distributed to \(\mG^T\mD\mG \mE_m\). I.e, if we let $\mG_m\in \R^{n\times m}$ contain the first $m$ columns of $\mG$ and let $\mG_{d-m}$ contain the remaining $d-m$ columns, our goal is the show that $\mM \eqdist \bmat{\mG_m ~ \mG_{d-m}}^T\mD\mG_m$.
	
	To see this is the case, let $\mZ\in \R^{n\times n}$ be a uniformly random orthogonal matrix and let $\mD_1, \mD_2$ be as in \problemref{trace-style}. The first observation is that $\mN\sim \cN(\mat{0},\mC)$ is identically distributed to  \(\mZ\mD\mS\) where $\mS\in \R^{d\times m}$ has standard normal entries $\sim \cN(0,1)$ and $\mD = \mD_1$ or $\mD = \mD_2$ with equal probability. This follows simply from that fact that $\mZ\mD_1 \mD_1\mZ^T = \mI$ and $\mZ\mD_2 \mD_2 \mZ^T = \mI - \vz_n\vz_n^T$, where $\vz_n$ is the last row of $\mZ$, which is a uniformly random unit vector.
	It follows that $\mN^T\mN \eqdist\mS^T\mD^T\mZ^T\mZ\mD\mS = \mS^T\mD\mS\eqdist \mG_m^T\mD\mG_m$.
	Next, observe that \(\mL^T\mZ\) is independent of $\mG$ and has i.i.d. \(\cN(0,1)\) entries since \mZ is orthogonal (and Gaussians are rotationally invariant). So, $\mL^T \mN \eqdist \mL^T \mZ\mD\mS \eqdist \mG_{d-m}^T\mD\mG$ and overall:
	\begin{align*}
	\mM
	&= \bmat{\mN^T\mN \\ \mL^T\mN}
	\eqdist \bmat{\mG_m^T\mD\mG_m \\ \mG_{d-m}^T\mD\mG_m}
	= \bmat{\mG_m ~ \mG_{d-m}}^T\mD\mG_m.
	\qedhere
	\end{align*}
\end{proof}

Finally, we directly prove a lower bound on the number of samples $m$ required to solve \problemref{statistical-estimation}, and thus, via \lemmaref{reduce-to-statistical-estimation}, \problemref{trace-style}. Combined with \lemmaref{lower-bound-trace-style}, this immediately yields our main lower bound on non-adaptive trace estimation, \theoremref{nonadaptive_lower_bound}.
\begin{lemma}
	\label{lem:statistical-lower-bound}
	If \(m < \frac{c}{\eps}\) for a fixed constant $c$, then \problemref{statistical-estimation} cannot be solved with probability \(\ge \frac23\).
\end{lemma}

\begin{proof}
	The proof follows from existing work on lower bounds for learning ``negatively spiked'' covariance matrices \cite{CaiMaWu:2015,PerryWeinBandeira:2018}. Let $\cP$ be the distribution of 
	$\mN$ in \problemref{statistical-estimation}, conditioned on $\mC = \mI$, and let $\cQ$ be the distribution conditioned on $\mC = \mI - \vz \vz^T$. These distributions fall into the spiked covariance model of \cite{PerryWeinBandeira:2018}, specifically the negatively spiked Wishart model (see Defn. 5.1 in \cite{PerryWeinBandeira:2018}) with spike size $\beta = -1$, and spike distribution $\cX$ the uniform distribution over unit vectors in $\R^n$. 
	Let $D_{\chi^2}(\cP\|\cQ)$ denote the $\chi^2$ divergence between $\cP$ and $\cQ$. Specifically,
	\begin{align*}
		D_{\chi^2}(\cQ\|\cP) = \int_{\mX \in \R_{d\times m}} \left(\frac{\cQ(\mX)}{\cP(\mX)}\right)^2\cP(\mX) d\mX - 1.
	\end{align*}
	We have $D_{KL}(\cQ\|\cP) \leq D_{\chi^2}(\cQ\|\cP)$, so to prove that $\cP,\cQ$ cannot be distinguished with good probability, it suffices to prove an upper bound on $D_{\chi^2}(\cQ\|\cP)$. In \cite{CaiMaWu:2015} (Lemma 7) it is proven that, letting $\vv$ and $\vv'$ be independent random unit vectors in $\R^n$,
	\begin{align}
		\label{eq:xisquare_bound}
		D_{\chi^2}(\cQ\|\cP) = \E_{\vv,\vv'}\left[\left(1-\langle \vv,\vv'\rangle^2\right)^{-m/2}\right] - 1.
	\end{align}
	Equation \eqref{eq:xisquare_bound} uses the notation of Prop. 5.11 in \cite{PerryWeinBandeira:2018}, which restates and proves a slightly less general form of the equality from \cite{CaiMaWu:2015}.
	Our goal is to prove that the expectation term in \eqref{eq:xisquare_bound} is $\leq 1 + C$ for some small constant $C$ when $m = \frac{c}{\epsilon} = cn$ for a sufficiently small constant $c$. 
	
	We first note that $\langle \vv,\vv'\rangle$ is identically distributed to $x\in[-1,1]$ where $x$ is the first entry in a random unit vector in $\R^n$. It is well known that 
	$\frac{x+1}{2}$ is distributed according to a beta distribution with parameters $\alpha = \beta = \frac{n-1}{2}$
	\cite{FangKotzNg:1990}. Specifically, this gives that $x$ has density:
	\begin{align*}
		p(x) = \frac{\Gamma(2\alpha)}{2\Gamma(\alpha)^2}\cdot\left(\frac{1-x^2}{4}\right)^{\alpha-1}.
	\end{align*} 
	Plugging this density back in to the expectation term in \eqref{eq:xisquare_bound} we obtain:	
	\begin{align*}
		\E_{\vv,\vv'}\left[\left(1-\langle \vv,\vv'\rangle^2\right)^{-m/2}\right] 
		&= \int_{-1}^1 \frac{\Gamma(2\alpha)}{2\Gamma(\alpha)^2}\cdot\left(\frac{1-x^2}{4}\right)^{\alpha-1} (1-x^2)^{-m/2} dx \\
		&= \int_{-1}^1 \frac{\Gamma(2\alpha)}{2\Gamma(\alpha)^2}\cdot\left(\frac{1}{4}\right)^{\alpha-1} (1-x^2)^{\alpha- 1- m/2} dx
	\end{align*}
	Assume without loss of generality that $n$ is an odd integer, and thus $\alpha = \frac{n-1}{2}$ is an integer.
	Let $m/2 = c\alpha$ for some constant $c\ll 1$ such that $c\alpha$ is an integer and thus $(1-c)\alpha$ is an integer. Then:
	
	\begin{align}
		\label{eq:exp_bound}
		\E_{\vv,\vv'}\left[\left(1-\langle \vv,\vv'\rangle^2\right)^{-m/2}\right]
		&= \frac{\frac{\Gamma(2\alpha)}{\Gamma(\alpha)^2}}{\frac{\Gamma(2(1-c)\alpha)}{\Gamma((1-c)\alpha)^2}}\cdot \left(\frac{1}{4}\right)^{c\alpha}\cdot \int_{-1}^1 \frac{\Gamma(2(1-c)\alpha)}{2\Gamma((1-c)\alpha)^2}\left(\frac{1}{4}\right)^{(1-c)\alpha-1}(1-x^2)^{(1-c)\alpha- 1} dx\\
		&= \frac{\frac{\Gamma(2\alpha)}{\Gamma(\alpha)^2}}{\frac{\Gamma(2(1-c)\alpha)}{\Gamma((1-c)\alpha)^2}}\cdot \left(\frac{1}{4}\right)^{c\alpha}, \nonumber
	\end{align}
	where the equality follows because the term being integrated is the density of $x$ where $\frac{x+1}{2}$ is distributed according to a beta distribution with parameters $(1-c)\alpha, (1-c)\alpha$.
	Since we have chosen parameters such that $\alpha$ is a positive integer, we have:
	\begin{align*}
		\frac{\Gamma(2\alpha)}{\Gamma(\alpha)^2} = \frac{(2\alpha -1)!}{(\alpha-1)!(\alpha-1)!} = \frac{\alpha}{2} \cdot \binom{2\alpha}{\alpha}.
	\end{align*} 
	Similarly, $\frac{\Gamma(2(1-c)\alpha)}{\Gamma((1-c)\alpha)^2} = \frac{(1-c)\alpha}{2} \cdot \binom{2(1-c)\alpha}{(1-c)\alpha}$. Each of the binomial coefficients in these expressions is a central binomial coefficient (i.e., proportional to a Catalan number), and we can use well known methods like Stirling's approximation to bound them. In particular, we employ a bound given in Lemma 7 of \cite{macwilliams1977theory}, which gives
$
		\frac{1}{2}\frac{4^z}{\sqrt{z}} \leq \binom{2z}{z}\leq \frac{1}{\sqrt{\pi}}\frac{4^z}{\sqrt{z}}. 
$ for any integer $z$.
	Accordingly, we have
	\begin{align*}
		\frac{\frac{\Gamma(2\alpha)}{\Gamma(\alpha)^2}}{\frac{\Gamma(2(1-c)\alpha)}{\Gamma((1-c)\alpha)^2}} &= \frac{1}{1-c} \cdot \frac{\binom{2\alpha}{\alpha}}{\binom{2(1-c)\alpha}{(1-c)\alpha}} \\
		&\leq \frac{1}{1-c} \cdot \frac{1/\sqrt{\pi}}{1/2}\frac{4^\alpha}{4^{(1-c)\alpha}} \cdot \sqrt{\frac{(1-c)\alpha}{\alpha}} \\
		&= \frac{2}{\sqrt{\pi (1-c)}} \cdot 4^{c\alpha}.
	\end{align*}
	Plugging into \eqref{eq:exp_bound} and requiring $c \le .1$ we have:
	\begin{align*}
		\E_{\vv,\vv'}\left[\left(1-\langle \vv,\vv'\rangle^2\right)^{-m/2}\right] \leq \frac{2}{\sqrt{\pi \cdot .9}} < \frac{6}{5}.
	\end{align*}
	It follows that $D_{KL}(\cQ\|\cP) \leq D_{\chi^2}(\cQ\|\cP) \leq \frac{6}{5}-1 = \frac{1}{5}$,
	and thus by Pinsker's inequality that 
	\begin{align*}
		D_{TV}(\cQ,\cP) \leq \frac{1}{\sqrt{10}} < \frac{1}{3}.
	\end{align*}
	Thus, no algorithm can solve \problemref{statistical-estimation} with probability $\ge \frac{1}{2} + \frac{1/3}{2} = \frac{2}{3}$, completing the lemma.
\end{proof}

\end{document}